\documentclass[12pt]{article}
\usepackage{amsmath}
\usepackage{amsthm}
\usepackage{graphicx}
\usepackage{enumerate}
\usepackage[numbers]{natbib}
\usepackage{bibunits}
\usepackage{url} 
\usepackage{multirow}
\newcommand{\blind}{1}

\usepackage{amssymb }
\usepackage{enumitem}
\usepackage{float}
\usepackage[ruled,vlined,linesnumbered]{algorithm2e}
\usepackage{chapterbib}
\usepackage[colorlinks=true, allcolors=blue]{hyperref}
\usepackage{titletoc}

\DeclareUnicodeCharacter{008D}{}
\DeclareUnicodeCharacter{0080}{}
\DeclareUnicodeCharacter{0093}{}

\newtheorem{theorem}{Theorem}[section]

\newtheorem{example}[theorem]{Example}

\newtheorem{remark}[theorem]{Remark}

\newtheorem{assumption}{Assumption}
\newtheorem{assumptionprime}{Assumption}

\newtheorem{lemma}[theorem]{Lemma}

\usepackage{setspace}
\usepackage{booktabs} 
\usepackage{caption} 
\usepackage{siunitx} 
\usepackage{threeparttable}
\usepackage{bbm}
\usepackage{xcolor}
\usepackage{sectsty}
\usepackage{listings}
\lstdefinestyle{promptstyle}{
  basicstyle=\ttfamily\small,
  breaklines=true,
  columns=fullflexible,
  frame=single,
  keepspaces=true,
  showstringspaces=false
}
\begin{document}


\def\spacingset#1{\renewcommand{\baselinestretch}%
{#1}\small\normalsize} \spacingset{1}


\if1\blind
{
  \title{\sf Empirical Likelihood with Generative AI}
  \author{ Jiguang Li \footnote{Jiguang Li is a doctoral student at the Booth School of Business, University of Chicago.},\,\, Sid Kankanala \footnote{Sid Kankanala is an Assistant Professor at the Booth School of Business, University of Chicago.} \,\,and Veronika Ro\v{c}kov\'a \footnote{Veronika Ročková is the Bruce Lindsay Professor in the Wallman Society of Fellows at the Booth School of Business, University of Chicago.}}
  \maketitle
} \fi

\if0\blind
{
   \title{\sf Empirical Likelihood with Generative AI}
  \maketitle
  \medskip
} \fi


\begin{abstract}
Moment conditions are widely used to identify parameters in models where the full likelihood is either unknown or intentionally left unspecified. Empirical likelihood methods address this problem by assigning probability weights to the observed data so that the sample moment conditions hold exactly. Building on this idea, we propose a nonparametric Bayesian framework based on exponentially tilted empirical likelihood. This Bayesian formulation is particularly appealing in settings where prior information is more naturally specified on the observables rather than on the underlying parameters. Such settings arise in the presence of auxiliary data sources or synthetic data generated by modern generative AI models.
Inference proceeds by projecting posterior draws from a Dirichlet process onto the moment-restricted model, yielding a computationally efficient procedure that is naturally amenable to parallelization.
We establish new Bernstein--von Mises and consistency theorems for the resulting projection posterior under both vanishing-prior and persistent-prior regimes. In an application to return prediction using overnight news headlines, we show that AI-generated auxiliary data can provide a useful source of indirect regularization when informative priors on the parameter itself are unavailable.
\end{abstract}

\noindent%
{\it Keywords:}  Bayesian nonparametrics, Bernstein-von Mises theorem, Bootstrap, Exponentially tilted empirical likelihood, Moment restrictions.  
\vfill
\vspace{-0.3cm}
\spacingset{1.8} 

\section{Introduction} \label{sec:intro}

In a large class of models, the parameter of interest is a finite-dimensional vector 
$\theta_0$ that satisfies a collection of moment restrictions:
\begin{equation} \label{eq:mc}
\mathbb{E}\!\left[g(x, \theta_0)\right] = 0.
\end{equation}
Here, $x \in \mathbb{R}^{d_x}$ denotes a random vector of observables, and 
$g(\cdot) = [g_1(\cdot), \dots, g_q(\cdot)]^\top$ is an $\mathbb{R}^q$-valued vector of known moment functions and $\theta_0 \in \Theta \subseteq \mathbb{R}^{d_{\theta}}$ is an unknown parameter of interest.  Such models are attractive because they do not require the researcher to fully specify the distribution of the data. Instead, identification is achieved through informative moment restrictions, often motivated by the structure of the application and the available sources of exogenous variation. In economics, common applications include models of consumer demand \cite{berry1995automobile,banks1997quadratic,blundell1994consumer}, firm productivity \cite{doraszelski2013r,boler2015r}, production functions \cite{ackerberg2015identification,levinsohn2003estimating}, dynamic panel data \cite{blundell1998initial,honore2025moment}, treatment effects \cite{chernozhukov2005iv,angrist1998children} and asset pricing \cite{hansen1982large,bansal1993no}.

In applications based on restrictions of the form in \eqref{eq:mc}, estimation has traditionally relied on two-step efficient generalized method of moments \cite{hansen1982large}, empirical likelihood \cite{10.1093/biomet/75.2.237,owen2001empirical}, and exponential tilting \cite{efron1981nonparametric,kitamura1997information}. While all three estimators are consistent and asymptotically efficient, empirical likelihood and exponential tilting are one-step procedures with distinct advantages: empirical likelihood has favorable higher-order asymptotic properties \cite{newey2004higher}, whereas exponential tilting is better behaved under model misspecification \cite{imbens2002generalized,imbens1995information}. Building on this complementarity, \cite{Schennach2007ETEL} introduced the exponentially tilted empirical likelihood (ETEL), which combines the higher-order refinements of empirical likelihood with the robustness of exponential tilting under misspecification. Additionally, it is shown that ETEL admits a nonparametric Bayesian interpretation in which the unknown data-generating distribution is integrated out as an infinite-dimensional nuisance parameter under a prior that favors entropy-maximizing weights \cite{f18ac6f4-7aa8-3e52-bc21-6455fbe9a264}.

Motivated by the attractive frequentist properties of ETEL, we develop a Bayesian framework in which posterior uncertainty about $\theta_0$ is obtained by filtering uncertainty about the data distribution through moment-restricted ETEL projections. This perspective is related in spirit to posterior-projection methods for constrained Bayesian inference, which project posterior draws onto a restricted space satisfying structural constraints \cite{LinDunson2014,ChakrabortyGhosal2022,AstfalckEtAl2026}. Here the projected object is different: we treat the unknown sampling distribution $F$ as the primitive object of inference, endow it with a nonparametric prior, and map each posterior draw of $F$ to $\theta^*(F)$ via a moment-restricted Kullback--Leibler (KL) projection induced by ETEL. Unlike classical ETEL, which uses the empirical distribution as its baseline, our approach defines the projection relative to a general discrete posterior draw of $F$, thereby allowing non-uniform weights over its support. This framework is especially appealing when the observed sample is limited but credible auxiliary information about the data-generating process is available.

To fix ideas, consider a Bayesian decision maker who is unable to specify a
tractable prior directly on $\theta_0$ but has access to a synthetic
auxiliary sample $\{x_j^*\}_{j=1}^m$, for example
obtained by repeatedly querying a large language model. We view the synthetic data as an approximate sample from a distribution $F_{\mathrm{AI}}$, which encodes indirect information about $\theta_0$ in the moment condition \eqref{eq:mc}. We formalize this using a Dirichlet process prior
$
F \sim \mathrm{DP}(\alpha, F_{\mathrm{AI}}),
$
where $\alpha > 0$ governs the strength of prior belief. Given an observed sample $\mathcal{D}_n = \{x_i\}_{i=1}^n$, the
posterior distribution (c.f. \cite{ferguson1974prior}) is
\begin{equation*} \label{eq:dp-posterior}
F \mid \mathcal{D}_n \sim \mathrm{DP}(\alpha + n,\, H_n), \qquad
H_n = \frac{\alpha}{\alpha + n}\,F_{\mathrm{AI}} +
\frac{n}{\alpha + n}\,\mathbb{P}_n,    
\end{equation*}
where $\mathbb{P}_n$ denotes the empirical distribution of
$\mathcal{D}_n$. The projection-based ETEL posterior for $\theta$ is then
obtained by pushing forward the posterior $F \mid \mathcal{D}_n$, that is, $\theta \mid \mathcal{D}_n \overset{d}{=} \theta^*(F)$.
Intuitively, auxiliary data from $F_{\mathrm{AI}}$ may regularize $\theta$ indirectly by encoding features of the data-generating process that are difficult to incorporate directly into a prior on $\theta$. In our applications, such information may include semantic patterns in news text, sector-specific language, or other distributional regularities that are informative about returns but difficult to encode directly in a prior on the structural parameter.

Our construction differs from the classical Bayesian ETEL (BETEL) approach of \cite{f18ac6f4-7aa8-3e52-bc21-6455fbe9a264,chib2018bayesian}, which treats the ETEL criterion as a likelihood-type object and combines it with a direct prior on $\theta$. This formulation is appealing when reliable prior information about $\theta$ is available, or when direct regularization of $\theta$ is important, and posterior computation is tractable. By contrast, we place a prior directly on $F$ and regularization of $\theta$ arises indirectly through the ETEL projection map $F \mapsto \theta^*(F)$. Part of our motivation for this approach is that, in many settings, prior
information may be more naturally formulated on the distribution of observables
than on the structural parameter $\theta$ itself; see, e.g.,
\cite{10.1214/aos/1176344611, ce5a0d81-15b4-3f2b-b252-6413a091240d}. This construction is also computationally attractive: Dirichlet process conjugacy makes posterior simulation for $F$ straightforward, while the ETEL projection can be computed independently across posterior draws.

Our approach to inference is closely related to recent work on parameters characterized as minimizers of expected loss functions \cite{10.1093/biomet/asz006,ohagan2025aipoweredbayesianinference}. The key distinction is that our parameter of interest is identified through general, possibly over-identified, nonlinear moment restrictions rather than through empirical risk minimization. This distinction is important because it allows our framework to accommodate a broad class of moment condition models that do not admit a natural empirical risk formulation. Our focus on ETEL is motivated by its interpretation as an efficient one-step estimator and by its favorable higher-order properties \cite{newey2004higher}, both of which are especially valuable when sample sizes are modest.

The main theoretical contributions of this paper are as follows. We
develop the inferential limit theory for the ETEL projection posterior of $\theta$. To obtain asymptotic approximations
that reflect the role of prior information in finite samples, we
consider two regimes: a vanishing-prior regime in which prior
influence is negligible relative to sampling uncertainty, and a persistent-prior
regime in which prior influence remains non-negligible. Under standard regularity conditions, we establish Bernstein-von Mises
(BvM) results for the posterior in both regimes. Our framework applies to the general over-identified setting in which
the number of moment conditions exceeds the dimension of the
parameter, and, to the best of our knowledge, provides the first BvM
result for the ETEL projection posterior induced by a nonparametric prior on $F$ in this generality. We further show that, in the exactly identified case, our procedure
nests the classical GMM Bayesian bootstrap of
\cite{Chamberlain01012003}, with the two procedures coinciding in
the limit as $\alpha \to 0$.  

We illustrate the practical viability of our procedure across a broad class of models, including asset pricing, average treatment effect estimation, demand estimation, and return prediction. Our generative AI experiments follow the increasingly popular trend of incorporating AI information in econometrics. Recent work has used generative AI to design in-silico experiments \cite{manning2024automated}, form economic expectations from historical news \cite{bybee2025ghost}, predict stock price reactions to news \cite{LopezLiraTang2023ChatGPTStocks, chen2022llmreturns}, and elicit investment preferences \cite{fedyk2024ai}, among other applications. Our results show that, when only moment restrictions are available, the Bayesian framework provides a natural way to incorporate AI-generated synthetic information as prior anchors within a model-free setting.

The paper is organized as follows. Section \ref{sec:etel} introduces the ETEL framework. Section \ref{sec:framework} describes our nonparametric Bayesian ETEL framework, with its asymptotic properties developed in Section \ref{sec:theory}. The empirical analysis has two parts. Section \ref{sec:betel-bootstrap} studies the baseline implementation, in which posterior uncertainty is driven by the observed sample and moment restrictions alone, using simulations and applications to asset pricing and average treatment effect estimation. Section \ref{sec:betel-ai} then turns to the AI-augmented implementation, in which large language models are used to construct a distributional prior for applications to financial-news return prediction and structural function recovery. Code for implementing our proposed algorithm and reproducing the experiments is available at
\url{https://github.com/JiguangLi/Empirical-Likelihood-with-Generative-AI}.

\section{Exponentially Tilted Empirical Likelihood} \label{sec:etel}

Moment restrictions in \eqref{eq:mc} provide a flexible way to learn about parameters of interest without specifying a full likelihood for the data-generating process. They encompass a broad class of econometric and statistical models in which identification is driven by structural moment conditions. Canonical examples include loss-based parameters, instrumental variables models, quantile regression, and dynamic panel models; see Appendix \ref{sec:motivating-examples}. We consider multiple contemporary examples later in Sections \ref{sec:betel-bootstrap}-\ref{sec:betel-ai}.

Exponentially tilted empirical likelihood (ETEL) procedures are widely used for inference in models identified by moment conditions; see, for example,
\cite{Schennach2007ETEL,kitamura2011bayesian,
f18ac6f4-7aa8-3e52-bc21-6455fbe9a264,chib2018bayesian,10.1111/rssb.12484, kim2026regularizedexponentiallytiltedempirical, 10.1093/biomet/asaa028,10.1111/rssb.12510,10.1111/j.1467-9868.2009.00715.x}. They provide a minimally invasive way to incorporate moment conditions into the empirical distribution by allowing unequal weights for individual observations. Unlike empirical likelihood, which selects weights by maximizing the nonparametric likelihood under moment restrictions \cite{10.1093/biomet/75.2.237,owen2001empirical, fccf7eeb-a52f-3454-95a0-7f88c15d64ea}, ETEL selects the feasible reweighting that is closest to the empirical distribution in KL divergence. Formally, given an observed sample $\mathcal{D}_n = \{x_i\}_{i=1}^n$, let $\mathbb{P}_n= n^{-1}\sum_{i=1}^n \delta_{x_i}$ denote the empirical distribution. For a fixed $\theta$, define:
\begin{equation*} \label{eq:prob-set}
    \mathcal{M}_{\theta}(\mathbb{P}_n)
:=
\Bigl\{
P:\;
P \ll \mathbb{P}_n,
\;
\int g(x,\theta)\,dP(x)=0
\Bigr\},
\end{equation*}
the class of probability measures defined on the same support as $\mathbb{P}_n$ that satisfy the sample moment restrictions. Here $P \ll \mathbb{P}_n$ means that $P$ may assign mass only to points in the support of $\mathbb{P}_n$. Hence $Q_{\theta} \in \mathcal{M}_{\theta}(\mathbb{P}_n)$ must take the form $Q_{\theta}=\sum_{i=1}^{n} w_i \delta_{x_i}$. ETEL finds the optimal reweighting $Q_{\theta}^*:= \sum_{i=1}^n w_{i}^*(\theta) \delta_{x_i}$ by projecting $\mathbb{P}_n$ to the moment restriction set $\mathcal{M}_{\theta}(\mathbb{P}_n)$:
\begin{equation} \label{eq:etel-inner}
    Q_\theta^*
\in
\arg\min_{Q_\theta \in \mathcal{M}_{\theta}(\mathbb{P}_n)}
D_{\mathrm{KL}}(Q_\theta \,\|\, \mathbb{P}_n).
\end{equation}
Since
\[
D_{\mathrm{KL}}(Q_\theta \,\|\, \mathbb{P}_n)
=
\sum_{i=1}^n w_i \log\!\left(\frac{w_i}{1/n}\right)
=
\log n + \sum_{i=1}^n w_i \log w_i,
\]
minimizing KL divergence is equivalent to maximizing the entropy of a discrete distribution defined through $\{w_i\}_{i=1}^n$. ETEL may thus be viewed as selecting the least informative reweighting of the empirical distribution that satisfies the sample moment restrictions. When $0$ lies in the interior of the convex hull of $\{g(x_i,\theta)\}_{i=1}^n$, standard Lagrange multiplier arguments following \cite{Schennach2007ETEL} yield exponentially tilted weights. The generalized weighted form used in our framework is given later in \eqref{eq:AI-SOL}.


Given the projection $Q_{\theta}^*:= \sum_{i=1}^n w_{i}^*(\theta) \delta_{x_i}$, one can construct the ETEL estimator using a ``maximum likelihood'' strategy, following \cite{Schennach2007ETEL}, by profiling these implied weights over $\theta$: 
\begin{equation}
\hat{\theta}_{\mathrm{ETEL}} = \arg\min_{\theta \in \Theta} \frac{1}{n} \sum_{i=1}^n - \log\left[n\,w_i^*(\theta)\right],
\label{eq:etel-estimator}
\end{equation}
which is equivalent to maximizing $\prod_{i=1}^n w_i^*(\theta)$. For our purpose, it is useful to rewrite it as
\begin{equation} \label{eq:etel-estimator-v2}
    \hat{\theta}_{\mathrm{ETEL}} = \arg \min_{\theta \in \Theta} \sum_{i=1}^n \frac{1}{n} \log\left[ \frac{1/n}{w_i^*(\theta)} \right]=\arg \min_{\theta \in \Theta} D_{\mathrm{KL}}(\mathbb{P}_n \,\|\, Q_\theta^* ).
\end{equation}
Thus, the standard ETEL estimator can also be viewed as minimizing another KL criterion. This re-interpretation of the ETEL estimator reveals that the ETEL procedure applies a KL projection twice, once for obtaining the weights $w_i^*(\theta)$ in \eqref{eq:etel-inner} and the second time for finding the estimator $\hat{\theta}_{\mathrm{ETEL}}$ in \eqref{eq:etel-estimator}. 

The ETEL framework is attractive for several reasons. First, it works
directly with the restrictions in \eqref{eq:mc}, without requiring a
fully specified likelihood for the data-generating process. Second, compared to two-step GMM, ETEL is an efficient one-step estimator and retains the favorable higher-order behavior of empirical likelihood
\cite{newey2004higher,Schennach2007ETEL}, making it particularly
appealing in small to moderate samples. Finally, ETEL behaves more robustly under model misspecification, in the sense that the associated estimator continues to target an
interpretable pseudo-true value and converges at the usual
$\sqrt{n}$ rate.\footnote{More precisely,
\cite{Schennach2007ETEL} shows that standard empirical likelihood can
fail to remain root-$n$ convergent under misspecification when the
moment functions are unbounded, whereas ETEL avoids this failure.}

Beyond strong frequentist properties, it is noteworthy that ETEL admits a specialized Bayesian interpretation. Specifically, \cite{f18ac6f4-7aa8-3e52-bc21-6455fbe9a264} show that
ETEL arises as the limit of a marginal posterior distribution obtained by integrating out an infinite-dimensional nuisance parameter under a particular conditionally i.i.d. likelihood representation and a carefully designed nonparametric prior. While this Bayesian motivation is intriguing, we do not rely on it here. Instead, we use the ETEL criterion as a likelihood-free link between the sampling distribution $F$ and the target parameter, and proceed with a nonparametric Bayesian analysis by treating the sampling distribution $F$ as random with a prior.

\section{Nonparametric Bayesian ETEL} \label{sec:framework}

We develop a Bayesian nonparametric framework by placing a prior on the sampling distribution rather than directly on $\theta$. The resulting posterior over sampling distributions induces a push-forward posterior on \(\theta\), naturally extending the KL projection and profiling steps that define the standard ETEL estimator.

\subsection{DP Posterior Computation} \label{pos-comp}

Perhaps the simplest way to place a nonparametric prior on the unknown
sampling distribution $F$ is through a Dirichlet process prior
\cite{ferguson1974prior}. Let $H_0$ be a baseline probability measure
and let $\alpha>0$ denote the concentration parameter. We model the
sampling distribution as $F \sim \mathrm{DP}(\alpha,H_0).$
Given an observed sample $\mathcal{D}_n=\{x_i\}_{i=1}^n$, conjugacy
gives the posterior $\mathrm{DP}(\alpha+n,H_n)$, where $H_n = \frac{\alpha}{\alpha+n} H_0 + \frac{n}{\alpha + n} \mathbb{P}_n$. If $H_0$ is continuous, we approximate the DP posterior by refreshing $m$ auxiliary atoms $\{x_j^*\}_{j=1}^m\sim H_0$ at each iteration and then applying Dirichlet reweighting. This implementation was deployed, for example, in \cite{fong2019scalablenonparametricsamplingmultimodal} as an alternative to the stick-breaking approximation \cite{Sethuraman1994}. The corresponding DP posterior computation can be found in lines $3$-$5$ of Algorithm \ref{alg:ai-betel}. If $H_0$ is atomic with finite support $\{x_j^*\}_{j=1}^m$, then the DP posterior draw can be represented exactly as a Dirichlet reweighting of the union of atoms $\{x_i\}_{i=1}^n$ and $\{x_j^*\}_{j=1}^m$ \cite{4dbd8cb8-8a01-32c4-b1eb-dd3774a40419}.

\begingroup
\setstretch{1.0}
\RestyleAlgo{ruled}
\SetKwInput{KwData}{Data}
\SetKwInput{KwHyper}{Hyperparameters}
\begin{algorithm}[!t]
\SetAlgoLined
\DontPrintSemicolon
\LinesNumbered
\caption{ETEL Posterior Projection Sampler}
\label{alg:ai-betel}

\KwData{Observed data $\{x_i\}_{i=1}^n$, continuous prior base measure $H_0$.}
\KwHyper{$\alpha \geq 0$ (DP concentration parameter), $B$ (number of posterior draws), $m$ (truncation level).}

Form the posterior base measure
$H_n=\frac{n}{n+\alpha}\mathbb{P}_n+\frac{\alpha}{n+\alpha}H_0$.

\For{$b=1,\ldots,B$}{
    Draw auxiliary atoms
    $x_1^{*(b)},\ldots,x_m^{*(b)}
    \stackrel{\mathrm{iid}}{\sim} H_0$.\;

    Draw weights
    $(v_1^{(b)}, \ldots, v_n^{(b)}, v_{n+1}^{*(b)},\ldots,v_{n+m}^{*(b)})
    \sim
    \mathrm{Dirichlet}\!\left(
    \underbrace{1,\ldots,1}_{n},
    \underbrace{\alpha/m,\ldots,\alpha/m}_{m}
    \right)$.\;

    Approximate the posterior by
    \[
    F^{(b)}
    =
    \sum_{i=1}^n v_i^{(b)}\,\delta_{x_i}
    +
    \sum_{j=1}^m v_{n+j}^{*(b)}\,\delta_{x_j^{*(b)}}.
    \]

    Solve the weighted ETEL problem \eqref{eq:AIEL}--\eqref{eq:AI-LOSS}
    with $F=F^{(b)}$ to obtain
    $\theta^{*(b)}=\theta^*(F^{(b)})$.\;
}
\KwRet{$\{\theta^{*(b)}\}_{b=1}^B$}
\end{algorithm}
\endgroup

\subsection{ETEL Projection Functional} \label{subsec:projection-functional}

Let $F^{(b)}$ be the $b^{th}$ DP posterior draw as described in Algorithm \ref{alg:ai-betel}. We may write  $$F^{(b)}=  \sum_{i=1}^n v_i^{(b)}\,\delta_{x_i}+ \sum_{j=1}^m v_{n+j}^{*(b)}\,\delta_{x_j^{*(b)}} = \sum_{k=1}^{K_b} \tilde{v}_k^{(b)} \delta_{\tilde{x}_k^{(b)}},$$
where $\{\tilde{x}_k^{(b)}\}_{k=1}^{K_b}$ are distinct support points of $F^{(b)}$. For simplicity, we suppress the dependence on $b$ and write $F^{(b)}=\sum_{k=1}^K v_k\,\delta_{x_k}$ below. To derive the induced ETEL projection functional $\theta^*(F^{(b)})$, we follow the same principles used to construct the standard ETEL estimator. Observe that the standard ETEL estimator can be viewed as a nested optimization problem: for a fixed $\theta$, the inner step finds the KL projection in \eqref{eq:etel-inner}, and the outer optimization finds $\hat{\theta}_{\mathrm{ETEL}}$ by minimizing the KL criterion derived in \eqref{eq:etel-estimator-v2}. 

We generalize this construction by defining $\theta^*(F^{(b)})$ through an analogous nested optimization problem. For fixed $\theta\in\Theta$, let:
$$\mathcal{M}_{\theta}(F^{(b)}) :=
\Bigl\{ P:\; P \ll F^{(b)}, \; \int g(x,\theta)\,dP(x)=0 \Bigr\},$$
the class of probability measures supported on $F^{(b)}$ that satisfy the moment restriction. The DP framework is particularly convenient here because the DP posterior draw $F^{(b)}$ is almost surely discrete, yielding a tractable constrained optimization. Next, define the KL projection of $F^{(b)}$ onto $\mathcal{M}_{\theta}(F^{(b)})$:
\begin{equation}\label{eq:projection}
P_{\theta}^{*}(F^{(b)})
\in
\arg\min_{P\in \mathcal{M}_{\theta}(F^{(b)})}
D_{\mathrm{KL}}(P\|F^{(b)}).
\end{equation}
Thus, the projection  amounts to finding the least informative reweighting of the posterior draw $F^{(b)}$ that satisfies the moment restrictions. This is the key conceptual departure from classical ETEL. In standard ETEL, the projection is taken relative to the fixed empirical distribution $\mathbb{P}_n$. Here, the reference distribution becomes $F^{(b)}$, which fluctuates around $H_n$ and carries non-uniform weights over its support.

Solving \eqref{eq:projection} follows from the same Lagrange multiplier
argument used in \cite{Schennach2007ETEL}. Any $P \in \mathcal{M}_{\theta}(F^{(b)})$ must take the form $P:= \sum_{k=1}^K p_k \delta_{x_k}$. Since $P \ll F^{(b)}$, we may equivalently write $p_k=v_k w_k$, where $w_k$ is the unknown multiplicative tilt applied to the baseline weight $v_k$. Given that
$$
D_{\mathrm{KL}}(P\|F^{(b)})
=
\sum_{k=1}^K p_k\log\frac{p_k}{v_k}
=
\sum_{k=1}^K v_k w_k\log w_k,
$$
$P_{\theta}^*(F^{(b)})$ can be characterized as the solution to a weighted entropy program
\begin{equation}\label{eq:AIEL}
\begin{aligned}
\max_{(w_1,\ldots,w_K)} \quad
&-\sum_{k=1}^K v_k w_k\log w_k \\
\text{subject to}\quad
&\sum_{k=1}^K v_k w_k = 1,
\qquad
\sum_{k=1}^K v_k w_k\,g(x_k,\theta)=0.
\end{aligned}
\end{equation}
If $0$ lies in the interior of the convex hull of the moment vectors evaluated at the support points of $F^{(b)}$, the optimal weights for each $k$ retain the exponential tilting form
\begin{equation}\label{eq:AI-SOL}
w_k^*(\theta;F^{(b)})
=
\frac{\mathrm{e}^{\,\lambda^*(\theta;F^{(b)})^\top g(x_k,\theta)}}
{\sum_{j=1}^{K} v_j\,\mathrm{e}^{\,\lambda^*(\theta;F^{(b)})^\top g(x_j,\theta)}},
\quad
\lambda^*(\theta;F^{(b)})
=
\arg\min_{\eta\in\mathbb{R}^q}
\log\!\Biggl[
\sum_{j=1}^{K} v_j\,\mathrm{e}^{\,\eta^\top g(x_j,\theta)}
\Biggr].
\end{equation}
It is worth noting that when $K=n$ and $v_k=1/n$, \eqref{eq:AI-SOL} reduces to the standard ETEL tilting weights \cite{Schennach2007ETEL}. The vector $\lambda^*(\theta;F^{(b)})$ is obtained by solving the convex dual problem, which can be done numerically using standard Newton methods.

After solving the inner problem, we proceed to formulate the outer problem that defines the ETEL functional. Recall from \eqref{eq:etel-estimator-v2} that the standard ETEL estimator can be written as $\hat{\theta}_{\mathrm{ETEL}} = \arg \min_{\theta \in \Theta} D_{\mathrm{KL}}(\mathbb{P}_n \,\|\, Q_\theta^* )$, where $Q_\theta^*$ is the solution to the inner ETEL problem. In our framework, the solution to the inner problem is $P_{\theta}^*(F^{(b)})$ and the reference measure is $F^{(b)}$ rather than $\mathbb{P}_n$. This leads to the corresponding ETEL criterion 
\[
\ell(\theta;F^{(b)})
=
D_{\mathrm{KL}}(F^{(b)} \,\|\, P_\theta^*(F^{(b)}))
=
\sum_{k=1}^K v_k \log \left[ \frac{v_k}{v_k w_k^*(\theta;F^{(b)})}\right]
=
-\sum_{k=1}^{K} v_k \log w_k^*(\theta;F^{(b)}).
\]
We then define the induced ETEL functional by
\begin{equation}\label{eq:AI-LOSS}
\theta^*(F^{(b)})
=
\arg\min_{\theta\in\Theta}
\ell(\theta;F^{(b)}).
\end{equation}

Relative to standard ETEL, which uses the fixed empirical distribution as its baseline, our procedure uses each DP posterior draw $F^{(b)}$ as the reference measure in the nested ETEL projection and then maps it to $\theta^{*(b)}=\theta^*(F^{(b)})$, as detailed in Algorithm \ref{alg:ai-betel}. It is precisely the variation of posterior draws $F^{(b)}$ that gives rise to the posterior variation of the parameter of interest $\theta$ through the nested KL projection represented by the push-forward map \eqref{eq:AI-LOSS}. Unlike standard BETEL procedures, Algorithm \ref{alg:ai-betel} allows for straightforward parallelization, since posterior draws and the associated nested optimization problems can be computed independently.

\subsection{ETEL bootstrap ($\alpha=0$) and AI-Powered ETEL ($\alpha >0$)}

When $\alpha =0$, Algorithm \ref{alg:ai-betel} reduces to an ETEL bootstrap, where we draw $(v_1, \ldots, v_n) \sim \mathrm{Dirichlet}(1, \ldots, 1)$ to form $F^{(b)}:= \sum_{i=1}^n v_i \delta_{x_i}$, analogous to the Bayesian bootstrap \cite{10.1214/aos/1176345338}. The push-forward draws $\theta^*(F^{(b)})$ are then obtained by solving the nested optimization problems \eqref{eq:AIEL}--\eqref{eq:AI-LOSS} as illustrated in Section \ref{subsec:projection-functional}. This setting is especially useful when no credible auxiliary information is available. In Section \ref{sec:betel-bootstrap}, we show that the resulting ETEL bootstrap procedure yields informative posterior inference in simulations, and in applications to asset pricing and average treatment effect estimation. 

When $\alpha>0$, treating $F$ as unknown with a $\mathrm{DP}(\alpha, H_0)$ prior invites the possibility of constructing an informative prior based on past data realizations or even generative AI data simulations. Related AI-induced nonparametric priors have been studied in loss-based Bayesian nonparametric inference \cite{ohagan2025aipoweredbayesianinference}. In our setting, the ETEL criterion in \eqref{eq:AI-LOSS} plays the role of the link between data and the
parameter of interest, but it is induced by moment restrictions rather than by a loss function. In Section \ref{sec:betel-ai}, we use large language models (LLMs) in two separate applications. In these applications, we set $H_0:= F_{\mathrm{AI}}$, and treat it as a continuous distribution, so the auxiliary atoms $\{x_j^{*(b)}\}_{j=1}^m$ are refreshed at each posterior draw through repeated prompting. 

The DP concentration parameter $\alpha>0$ has a natural interpretation as an effective prior sample size and is the primary tuning parameter. By contrast, $m$ is only the truncation level used to approximate a continuous DP posterior draw, so it should be chosen sufficiently large in practice. Following the $\alpha$-calibration discussion in \cite{ohagan2025aipoweredbayesianinference}, one option is coverage matching: for each candidate $\alpha$, bootstrap datasets from the empirical distribution, recompute the posterior credible region under Algorithm \ref{alg:ai-betel}, and select the largest $\alpha$ such that the nominal $(1-\delta)$ credible region contains the corresponding standard ETEL estimator computed on the same bootstrap sample with frequency at least $(1-\delta)$. A second option is asymptotic covariance matching: one may choose $\alpha$ so that a scalar summary of the empirical posterior covariance under Algorithm \ref{alg:ai-betel}, such as its trace or average marginal variance, remains close to the standard ETEL benchmark $J_0^{-1}/n$ from Theorem \ref{thm:bootstrap}.

\subsection{Regularization Induced by $F_{\rm{AI}}$} \label{subsec:reg}
Although the prior in our framework is placed on the sampling distribution, it induces a push-forward prior on $\theta$. Even a diffuse prior on $F$ will ultimately result in a somewhat informative prior on $\theta$, since the push-forward map is  driven by the underlying moment restrictions. Let $\Pi_F(\cdot)$ denote a prior on $F$, then the induced prior on $\theta$ is $\Pi_{\theta}(A) = \Pi_{F}(\{F: \theta^*(F) \in A\})$. This prior is generally not available in closed form. In linear regression, however, the induced regularization has a transparent representation.

Consider the moment condition $g(z,\beta)=x(y-x^\top\beta)$, where $z=(y,x)$, $x\in\mathbb{R}^{p}$, and $\beta\in\mathbb{R}^{p}$. Under the prior $\mathrm{DP}(\alpha,F_{\rm AI})$,  the DP posterior follows $\mathrm{DP}(\alpha+n, H_n)$, where $H_n= \frac{n}{n+\alpha} \mathbb{P}_n + \frac{\alpha}{n+\alpha}F_{\rm AI}$. In this exactly identified setting, the ETEL projection $\beta(F)$ coincides with the unique moment root:
\begin{equation*} \label{eq:beta-f}
    \beta(F):= S(F)^{-1}t(F), \qquad S(F):= \mathbb{E}_{F}[xx^\top], \qquad t(F):= \mathbb{E}_{F}[xy],
\end{equation*}
provided these moments exist and $S(F)$ is invertible almost surely.

The functional $\beta(H_n)$ provides a central summary of the push-forward posterior. Under squared prediction loss, it minimizes the posterior prediction risk. Locally, $\beta(H_n)$ only differs from the posterior mean $\mathbb E\{\beta(F)\mid\mathcal{D}_n\}$ by second-order terms.\footnote{A first-order Taylor expansion around $\beta(H_n)$ yields $\beta(F)-\beta(H_n) \approx S(H_n)^{-1}\Delta_t \;-\; S(H_n)^{-1}\Delta_S \,\beta(H_n),$ where $\Delta_{S}:= S(F) - S(H_n)$ and $\Delta_{t}:= t(F) - t(H_n)$.
Since $\mathbb{E}[\Delta_S \mid \mathcal{D}_n]=0$ and $\mathbb{E}[\Delta_t \mid \mathcal{D}_n]=0$, the right-hand side vanishes after conditioning on $\mathcal{D}_n$. } The moment condition $\mathbb{E}_{H_n}[g(z,\beta)]=0$ admits an explicit finite-sample form:
\[
\sum_{i=1}^n x_i(y_i-x_i^\top\beta) \;+\; \alpha\,\mathbb{E}_{F_{\rm AI}}\!\left[x(y-x^\top\beta)\right]=0.
\]
Define $S_{\rm AI}:=\mathbb{E}_{F_{\rm AI}}[xx^\top]$, $t_{\rm AI}:=\mathbb{E}_{F_{\rm AI}}[xy]$, and $\beta_{\rm AI}:=S_{\rm AI}^{-1}t_{\rm AI}$. By definition,  $\mathbb{E}_{F_{\rm AI}} \left[x(y-x^\top\beta)\right]$ can be rewritten as $S_{\rm AI}(\beta_{\rm AI}-\beta)$. Hence we can express $\beta(H_n)$ in matrix form as
\begin{equation*} \label{eq:gr-estimator}
\beta(H_n) = \bigl(X^\top X + \alpha S_{\rm AI}\bigr)^{-1}\bigl(X^\top Y + \alpha S_{\rm AI}\beta_{\rm AI}\bigr),
\end{equation*}
which can be viewed as the minimizer of a generalized ridge regularization problem \cite{69896607-a049-30fc-992a-9822d3c5f921}:
\begin{equation} \label{eq:generalized-ridge-obj}
\beta(H_n) = \arg\min_{\beta\in\mathbb{R}^p}
\left\{
\|Y-X\beta\|^2 + \alpha (\beta-\beta_{\rm AI})^\top S_{\rm AI}(\beta-\beta_{\rm AI})
\right\}.
\end{equation}
Equation \eqref{eq:generalized-ridge-obj} makes the regularization mechanism explicit: the concentration parameter $\alpha$ controls the strength of shrinkage, $S_{\rm AI}$ determines the geometry of the penalty, and $\beta_{\rm AI}$ is the shrinkage target induced by $F_{\rm AI}$ through its best linear predictor. If $F_{\rm AI}$ is close to the true sampling distribution, then $\beta_{\rm AI}$ is close to the true regression coefficient, so the regularization shrinks toward a scientifically meaningful target rather than toward zero. Standard ridge regression is recovered as the special case $S_{\rm AI}=I_p$ and $t_{\rm AI}=0$.

This calculation also clarifies the relation to imaginary-data and catalytic priors \cite{doi:10.1073/pnas.1920913117}. These approaches
regularize parameters through synthetic observations or an explicit
parametric prior, such as Zellner's $g$-prior
\cite{Zellner1986GPrior}. Here the regularization arises instead from a
nonparametric prior on the sampling distribution. We emphasize that the induced prior on $\beta$ need not be normal, since the generalized ridge form describes the posterior predictive center $\beta(H_n)$ induced by $F_{\rm AI}$.

\section{Asymptotic Normality} \label{sec:theory}
We study the frequentist properties of the posterior distribution induced by our ETEL projection sampler. The central question is whether replacing the uniform empirical reference measure by a DP posterior draw changes the first-order behavior of ETEL, and, if so, how that change depends on the strength of the AI prior. We provide answers in two regimes. When the prior is asymptotically negligible, the induced posterior is asymptotically first-order equivalent to the standard Bayesian ETEL posterior \cite{10.1111/rssb.12484}. When the prior carries non-vanishing mass, the posterior remains asymptotically Gaussian, but is centered at the pseudo-true solution associated with the mixture law of $F_{\rm{AI}}$ and $F_0$.

It is tempting to interpret our ETEL criterion as a weighted likelihood bootstrap (WLB) objective \cite{10.1111/j.2517-6161.1994.tb01956.x}, which maximizes a weighted log-likelihood $\arg \max_{\theta} \sum_{i=1}^n v_i \log p(x_i|\theta)$ in a parametric model. Although WLB draws are first-order equivalent to the Bayesian posterior in a well-specified parametric model, \footnote{Let $\hat{\theta}_n$ be the maximum likelihood estimator, $\tilde{\theta}_n$ a random WLB draw. First-order equivalence means that, conditional on the data, $\sqrt{n}(\theta-\hat{\theta}_n)$ and $\sqrt{n}(\tilde{\theta}-\hat{\theta}_n)$ converge to the same limit.}  our criterion is not the log-likelihood of any fixed parametric model, so the standard WLB theory does not directly apply. Our approach is closer in spirit to loss-likelihood bootstrap (LLB) \cite{10.1093/biomet/asz006}, which also puts a prior on the DGP directly. In LLB, the functional $\theta(F)$ is defined through a loss function $l(\theta, x)$ such that $\theta(F) = \arg \min_{\theta \in \Theta} \int l(\theta, x) dF(x)$. Their conditional CLT result relies on an additive empirical risk structure, where the objective is a sum of i.i.d. terms with a fixed loss $l(\theta,x)$ that does not itself depend on $F$. By contrast, our ETEL criterion in \eqref{eq:AI-LOSS} is not additive in this sense, since the integrand $-\log \frac{dP_{\theta}^*(F)}{dF}$ itself depends on $F$. 

Let $F_0$ denote the true sampling distribution, and $\Pi_n(\cdot\mid \mathcal{D}_n)$ denote the push-forward posterior distribution of $\theta^*$ generated by Algorithm \ref{alg:ai-betel}. Our asymptotic results assume a continuous base measure $F_{\rm{AI}}$ and rely on the following assumptions.

\begin{assumption}[Identification]\label{A1}
The parameter space $\Theta$ is compact. The mapping $\bar{g}(\theta):= \mathbb{E}_{F_0}[g(x,\theta)]$ is continuous on $\Theta$ and has a unique zero at $\theta_0$. Define
$$
\Omega(\theta):= \mathbb E_{F_0}[g(x,\theta)g(x,\theta)^\top] -\bar g(\theta)\bar g(\theta)^\top,
\qquad
G(\theta):=\mathbb E_{F_0}[\nabla_\theta g(x,\theta)].
$$
The matrix $\Omega_0 := \Omega(\theta_0)$ is positive definite, and $G_0 := G(\theta_0)$ has full column rank $d_{\theta}$. There exists a compact convex neighborhood $\mathcal{N} \subset \mathrm{int}(\Theta)$ of $\theta_0$ such that $\inf_{\theta\in\mathcal N}\lambda_{\min}\{\Omega(\theta)\}>0$ and $\inf_{\theta\in\mathcal N}
\lambda_{\min}\{G(\theta)^\top\Omega(\theta)^{-1}G(\theta)\}>0$.
\end{assumption}

\begin{assumption}[Smoothness and integrability]\label{A2}
There exists $\delta >0$ such that:
\begin{enumerate}
\item[(i)] For every $x$, $g(x,\theta)$ is continuous on $\Theta$ and twice continuously differentiable on $\mathcal N$. There exists an envelope $M_\Theta(x):=\sup_{\theta\in\Theta}\|g(x,\theta)\|$ such that $\mathbb E_{F_0}[M_\Theta(x)^{2+\delta}]<\infty.$

\item[(ii)] There exist local derivative envelopes $L_1(x):=\sup_{\theta\in\mathcal N}\|\nabla_\theta g(x,\theta)\|$ and $L_2(x):=\sup_{\theta\in\mathcal N}\|\nabla_\theta^2 g(x,\theta)\|$ that satisfy $\mathbb E_{F_0}[L_1(x)^2]<\infty$, and $\mathbb E_{F_0}[L_2(x)]<\infty$.

\item[(iii)] There exists a compact convex set $\Lambda \subset \mathbb{R}^q$ with $0 \in \mathrm{int}(\Lambda)$ such that $$\mathbb E_{F_0} \left[
\sup_{\theta\in\Theta,\ \eta\in\Lambda} e^{\eta^\top g(x,\theta)} (1+\|g(x,\theta)\|^3)\right]<\infty.$$

\item[(iv)] When $\alpha_n >0$, the same bounds in parts (i)-(iii) hold with $F_0$ replaced by $F_{\mathrm{AI}}$.
\end{enumerate}
\end{assumption}

\begin{assumption}[Feasibility]\label{A3}
With $\psi_0(\eta, \theta):= \log \mathbb{E}_{F_0} \left [e^{\eta^\top g(x,\theta)}\right]$, $\lambda_0(\theta):= \arg \min _{\eta \in \Lambda} \psi_0(\eta, \theta)$ exists, is unique, and lies in $\mathrm{int}(\Lambda)$ for every $\theta \in \Theta$. With probability $1-o(1)$, the sample dual minimizer  $\arg\min_{\eta\in\Lambda} \log\!\left [
\sum_{j=1}^{n+m} v_j\,\mathrm{e}^{\,\eta^\top g(x_j,\theta)} \right]$ exists, is unique, and lies in $\mathrm{int}(\Lambda)$ for every $\theta \in \Theta$.
\end{assumption}

\begin{assumption}[AI prior]\label{A4}
We consider either (i) (ETEL bootstrap): $\alpha_n=0$ and $m=0$, or (ii) (AI-augmented prior) $\alpha_n=o(\sqrt n)$. 
\end{assumption}

We first show that posterior consistency holds under relatively weak conditions.

\begin{theorem}[Posterior consistency]\label{thm:posterior-consistency}
Suppose the compactness and unique-identification conditions in Assumption \ref{A1} hold, and Assumption \ref{A3} holds. Suppose also that $g(x,\theta)$ is continuous in $\theta$ for every $x$, and that $\sup_{\theta\in\Theta}\|g(x,\theta)\|$ and $\sup_{\theta\in\Theta,\eta\in\Lambda} \left [e^{\eta^\top g(X,\theta)}\right]$
are integrable under $F=F_0$, and $F=F_{\rm AI}$ when $\alpha_n>0$.  If
$\alpha_n/n\to0$, then, for every $\varepsilon>0$,
$$ \Pi_n \bigl( \|\theta^*-\theta_0\|>\varepsilon \mid \mathcal{D}_n
\bigr) \xrightarrow{\mathbb P}0 .
$$
\end{theorem}

Let $\overset{\mathbb{P}}{\rightsquigarrow} $ denote weak convergence of the conditional posterior distribution in probability. The following theorem shows that, in the absence of $F_{\rm{AI}}$ or when it is asymptotically negligible, our projection-based ETEL posterior draw concentrates around the standard ETEL estimator $\hat{\theta}_n$ at the nominal $\sqrt{n}$ rate with the standard ETEL information matrix. 

\begin{theorem}[Gaussian limit under a vanishing AI prior]\label{thm:bootstrap}
Let $\hat{\theta}_n$ be the standard ETEL estimator defined in (\ref{eq:etel-estimator}), and  $\theta^*$ be a generic posterior draw from Algorithm \ref{alg:ai-betel}. Under Assumptions \ref{A1}-\ref{A4}, 
\begin{equation*} \label{eq:b-clt}
  \sqrt{n} (\theta^* - \hat{\theta}_n) \mid \mathcal{D}_n \overset{\mathbb{P}}{\rightsquigarrow} \mathcal{N}(0, J_0^{-1}), \quad J_0:= G_0^{\top} \Omega_0^{-1} G_0.  
\end{equation*}
\end{theorem}

\begin{remark}[Frequentist validity]
Let $\hat J_n$ be a consistent estimator of $J_0$, and $q_{n,1-\tau}^*$ be the conditional $(1-\tau)$-quantile of $n(\theta^*-\hat\theta_n)^\top \hat J_n(\theta^*-\hat\theta_n)$ given $\mathcal{D}_n$. Define
$$C_{n,1-\tau} := \left\{\theta\in\Theta: n(\theta-\hat\theta_n)^\top \hat J_n(\theta - \hat\theta_n) \le q_{n,1-\tau}^*\right\}.$$
Theorem \ref{thm:bootstrap} shows that $C_{n,1-\tau}$ is a valid confidence set with $\mathbb P_{F_0}\{\theta_0\in C_{n,1-\tau}\}\to 1-\tau.$
\end{remark}

\subsection{Connection to Bayesian GMM Bootstrap} \label{sec:cbb}

In the exactly identified case, $q=d_\theta$, Theorem \ref{thm:bootstrap} also provides theoretical justification for the Bayesian bootstrap procedure of \cite{Chamberlain01012003}. To see this, consider a given bootstrap draw $v \sim \text{Dirichlet}(1, \ldots, 1)$ and define the weighted sample moments $\bar{g}_v(\theta):= \sum_{i=1}^n v_i g(x_i, \theta)$. \citet{Chamberlain01012003} compute $\hat\theta_{\rm CI}(v)$ for each
Dirichlet draw by solving:
\begin{equation} \label{eq:ci}
   \sum_{i=1}^n v_i g(x_i, \hat{\theta}_{\rm{CI}}(v)) =0.
\end{equation}
Theorem \ref{thm:cbb} shows that, in this exactly identified case without $F_{\rm{AI}}$, our projection ETEL sampler for a fixed Dirichlet draw $v$ coincides with $\hat\theta_{\rm CI}(v)$.

\begin{theorem}[Connection to Bayesian bootstrap] \label{thm:cbb}
Consider an exactly identified case and fix a given Dirichlet draw $v \sim \text{Dirichlet} (1, \ldots, 1)$. Suppose there exists a unique $\theta^\dagger \in \Theta$ such that $\bar{g}_v(\theta^\dagger):= \sum_{i=1}^n v_i g(x_i, \theta^\dagger)=0$, and the weighted entropy program in (\ref{eq:AIEL})  is feasible. If $\theta^*(v)$ is the minimizer of the ETEL criterion from Algorithm \ref{alg:ai-betel} with $\alpha=0$,
$$\hat{\theta}_{\rm{CI}}(v) = \theta^{\dagger}= \theta^*(v).$$
\end{theorem}

The proof of Theorem \ref{thm:cbb} is provided in Appendix \ref{subsec:e-proof}. This theorem provides further intuition in the exact-identification case: the ETEL projection step $P_{\theta}^*(F^{(b)})$ is simply $F^{(b)}$, since we can find corresponding $\theta^*(F^{(b)})$ by solving the moment equations without tilting the weights. This is practically useful, since we can avoid solving the nested optimization in Algorithm \ref{alg:ai-betel} by solving a simpler weighted moment equation instead. In general, however, and in particular in the over-identified case ($q>d_{\theta}$), there typically does not exist any $\theta$ such that $\sum_i v_i g(x_i, \theta)=0$ holds exactly. In that regime, even without AI augmentation, our method enforces the moment restrictions by exponentially tilting the weights $w_i^*(\theta)$ and does not coincide with the GMM Bayesian bootstrap procedure.

\subsection{Non-vanishing AI Prior}

In many applications, however, one may wish to assign non-vanishing prior
mass to the auxiliary distribution. We formalize this regime as follows.

\begin{assumptionprime}[Persistent prior]\label{A4prime}
For a fixed $\gamma \in (0, \infty)$, let $\alpha_n = \gamma n$ and $m_n \rightarrow \infty$. 
\end{assumptionprime}
Define the mixed law $F_{\gamma}:= (1-\delta_{\gamma}) F_0 + \delta_{\gamma} F_{\rm{AI}}$ with $\delta_{\gamma}:= \frac{\gamma}{1+\gamma}$. Let $\theta_{\gamma}$ be the target parameter under $F_{\gamma}$ with the corresponding estimator defined as $\hat{\theta}_{n,\gamma}:= \arg \min_{\theta \in \Theta} \ell(\theta; F_{n,\gamma})$, where $F_{n,\gamma}:= \frac{n}{n+\alpha_n} \mathbb{P}_n + \frac{\alpha_n}{n+\alpha_n} \mathbb{P}_{m_n}^*$ with $\mathbb{P}_{m_n}^* := \frac{1}{m_n} \sum_{j=1}^{m_n} \delta_{x_j^*}$. Additionally, we define 
$$\mu_{\gamma}(\theta):= \mathbb{E}_{F_{\gamma}}[g(x, \theta)],\quad \Omega_{\gamma}(\theta):= \mathbb{E}_{F_{\gamma}}[g(x,\theta)g(x,\theta)^\top]-\mu_{\gamma}(\theta) \mu_{\gamma}(\theta)^\top, \quad G_{\gamma}(\theta):= \mathbb{E}_{F_{\gamma}}[\nabla_{\theta} g(x,\theta)].$$

\begin{theorem}[Gaussian limit under a persistent AI prior]\label{thm:b-non-vanishing}
Suppose Assumption \ref{A4prime} holds, and the analogues of Assumptions \ref{A1}--\ref{A3} hold with $(F_0,\theta_0,\bar g,\Omega,G)$ replaced by
$(F_\gamma,\theta_\gamma,\mu_\gamma,\Omega_\gamma,G_\gamma)$. Then, conditional on the augmented data $\mathcal{D}_{n,m_n}:= \{x_1, \ldots, x_n, x_1^*, \ldots, x_{m_n}^*\}$,
\begin{equation*} \label{eq:ai-clt}
    \sqrt{n+\alpha_n} (\theta^*(V_n) - \hat{\theta}_{n,\gamma}) \mid \mathcal{D}_{n,m_n} \overset{\mathbb{P}}{\rightsquigarrow} \mathcal{N}(0, J_\gamma^{-1}), \quad J_\gamma:= G_{\gamma,0}^{\top} \Omega_{\gamma,0}^{-1} G_{\gamma,0},
\end{equation*}
where $G_{\gamma,0}:=G_\gamma(\theta_\gamma)$ and $\Omega_{\gamma,0}:=\Omega_\gamma(\theta_\gamma)$.
\end{theorem}

Theorem \ref{thm:b-non-vanishing} characterizes the first-order behavior under a persistent AI prior. When $\alpha_n=\gamma n$, the posterior is centered at the ETEL solution based on the empirical mixture, which converges to the mixture law $F_\gamma$. Thus the population target is $\theta_\gamma$ rather than, in general, $\theta_0$. If $F_{\rm AI}$ is substantially different from $F_0$, then $\theta_\gamma$ may differ from the original target $\theta_0$, so the posterior targets a prior-shifted parameter. When $F_{\rm AI}$ is well aligned with $F_0$, the additional prior mass can reduce posterior dispersion through the effective sample size $n+\alpha_n$. Taken together, Theorems \ref{thm:bootstrap} and \ref{thm:b-non-vanishing} show how the induced ETEL posterior changes as the prior mass ranges from asymptotically negligible to persistent. The proofs for these two theorems are provided in Appendix \ref{sec:bvm-proofs}.

\section{ETEL bootstrap $(\alpha=0)$} \label{sec:betel-bootstrap}

For all simulation and real-data experiments in this section, we set $\alpha=0$ and do not incorporate any generative-model information. The parametric examples considered here show that sensible Bayesian inference can arise through our projection-based ETEL posterior. We illustrate this by revisiting several examples from \cite{10.1111/rssb.12484}, which conducts inference using a Bayesian exponentially tilted empirical likelihood with a direct prior on the parameter. Computation also remains naturally parallelizable: as summarized in Algorithm \ref{alg:ai-betel}, the ETEL projection can be carried out independently for each posterior draw.

\subsection{Over-identified linear IV} \label{subsec:oiv}

We consider a simple over-identified linear IV design. Let $z_i=(z_{1i},z_{2i})^\top$ denote a vector of instruments, independent of the errors $(u_i,v_i)$. The data are generated by
$$
y_i = \beta_0 x_i + u_i,  \qquad  x_i = \pi_1 z_{1i} + \pi_2 z_{2i} + v_i,
$$
with
\[
z_i \sim \mathcal N(0,I_2), 
\qquad 
(u_i,v_i)^\top \sim \mathcal N(0,\Sigma_{uv}), 
\qquad
\Sigma_{uv} =
\begin{pmatrix}
1 & \rho_{uv} \\
\rho_{uv} & 1
\end{pmatrix}.
\]
In the simulation, we generate $n=500$ observations, with $\beta_0=1$, $\pi_1=0.8$, $\pi_2=0.6$, and $\rho_{uv}=0.4$. Since $v_i$ is correlated with $u_i$, the regressor $x_i$ is endogenous. The population OLS estimand is $\beta_0+\frac{\operatorname{Cov}(x_i,u_i)}{\operatorname{Var}(x_i)} =1.2,$ so the sample OLS estimate is expected to be biased away from $\beta_0$. In this setting, $\beta_0$ is identified by the moment restrictions
\[
\mathbb{E}[z_{1i}(y_i-x_i\beta_0)]=0,
\qquad
\mathbb{E}[z_{2i}(y_i-x_i\beta_0)]=0.
\]

We fit the baseline ETEL projection sampler on this simulated dataset with $10,000$ draws. Table \ref{tab:oiv-summary} reports posterior summaries for $\beta$. The posterior is centered at the true value $\beta_0=1$, whereas the sample OLS estimate is close to $1.2$ as visualized in Figure \ref{fig:o-iv}. This simple example illustrates that the projection sampler recovers the moment-identified IV coefficient even when the naive regression target is distorted by endogeneity.

\begin{table}[t]
\centering
\caption{ETEL bootstrap posterior summary in the over-identified linear IV model.}
\label{tab:oiv-summary}
\begin{tabular}{lcrrrrr}
\toprule
Method & Parameter & Mean & SD & Median & Lower & Upper \\
\midrule
ETEL bootstrap & $\beta$ & 1.000 & 0.043 & 1.000 & 0.928 & 1.069 \\
\bottomrule
\end{tabular}
\vspace{0.3em}

\begin{minipage}{0.95\linewidth}
\footnotesize
\emph{Notes:} Results are based on $10{,}000$ posterior draws with
$\alpha=0$. ``Lower'' and ``Upper'' denote the $0.05$ and $0.95$
posterior quantiles. The true value is $\beta_0=1$.
\end{minipage}
\end{table}

\begin{figure}[t]          
  \centering
  \includegraphics[width=0.7\textwidth]{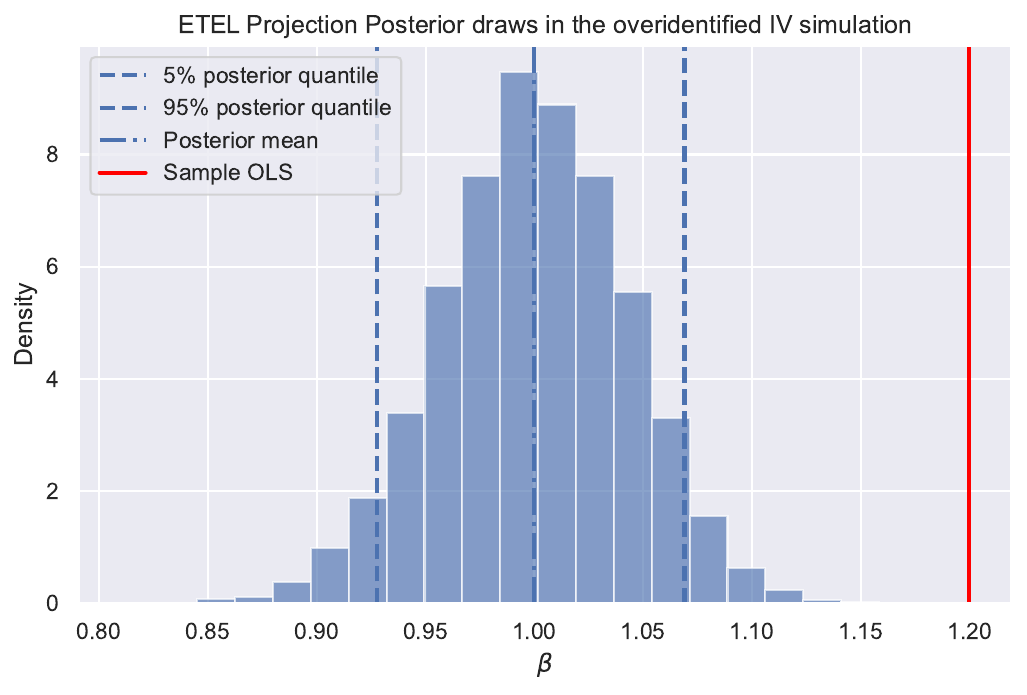} 
  \caption{Posterior draws visualization in the over-identified linear IV simulation.}
  \label{fig:o-iv}
\end{figure}

\subsection{Conditional Heteroscedasticity}

We consider the simulation design in Example 1 of \cite{10.1111/rssb.12484}. 
We generate $n=250$ synthetic observations $\{(y_i,x_i)\}_{i=1}^n$. The covariate $x_i$ is drawn independently from a uniform distribution on $[-1,\,2.5]$, and the true
regression coefficients are set to $\theta_0=1$ and $\theta_1=1$. The
outcome variable is generated according to
\[
y_i = \theta_0 + \theta_1 x_i + \varepsilon_i, \quad \varepsilon_i \sim \mathrm{SN}\bigl(m(x_i),\,h(x_i),\,s(x_i)\bigr)
\]
where the error term $\varepsilon_i$ follows a skew-normal distribution, with location, scale, and shape parameters depending on $x_i$. The location function is defined as $m(x) = -\,h(x)\sqrt{\tfrac{2}{\pi}}\,
\frac{s(x)}{\sqrt{1+s(x)^2}}$, which ensures that the conditional mean restriction
$\mathbb{E}[\varepsilon_i \mid x]=0$ holds. The scale and shape functions
are specified as $h(x) = \sqrt{\exp(1 + 0.7x + 0.2x^2)}$ and $s(x) = 1 + x^2$.

Inference is based exclusively on the conditional moment
restriction $\mathbb{E}[\varepsilon_i \mid w_i]=0$, where $w_i=(1,x_i)^\top$. Following \cite{CHIB2010322}, the conditional moment restriction is transformed into a set of unconditional moment conditions by interacting the residual $\varepsilon_i$ with a $K$-dimensional sieve basis constructed from natural cubic splines. We follow their construction of the basis matrix, including column-differencing to avoid redundancy. Posterior inference is conducted for 
$K \in \{3, 5,\,10,\,15,\,20\}$.

\begin{table}[htbp]
\centering
\caption{ETEL bootstrap posterior summaries under different sieve dimensions $K$.}
\label{tab:posterior-summaries-K}
\begin{tabular}{ccrrrrr}
\toprule
$K$ & Parameter & Mean & SD & Median & Lower & Upper \\
\midrule
\multirow{2}{*}{$3$}
& $\theta_0$ & 1.011 & 0.095 & 1.010 & 0.858 & 1.169 \\
& $\theta_1$ & 1.043 & 0.130 & 1.040 & 0.837 & 1.258 \\
\midrule
\multirow{2}{*}{$5$}
& $\theta_0$ & 1.010 & 0.094 & 1.010 & 0.858 & 1.166 \\
& $\theta_1$ & 1.023 & 0.118 & 1.023 & 0.832 & 1.221 \\
\midrule
\multirow{2}{*}{$10$}
& $\theta_0$ & 0.992 & 0.093 & 0.992 & 0.841 & 1.144 \\
& $\theta_1$ & 1.001 & 0.113 & 1.001 & 0.816 & 1.187 \\
\midrule
\multirow{2}{*}{$15$}
& $\theta_0$ & 0.986 & 0.092 & 0.986 & 0.836 & 1.136 \\
& $\theta_1$ & 0.993 & 0.114 & 0.993 & 0.806 & 1.179 \\
\midrule
\multirow{2}{*}{$20$}
& $\theta_0$ & 0.954 & 0.091 & 0.954 & 0.803 & 1.103 \\
& $\theta_1$ & 0.920 & 0.108 & 0.920 & 0.744 & 1.100 \\
\bottomrule
\end{tabular}

\vspace{0.5em}
\begin{minipage}{0.95\linewidth}
\footnotesize
\emph{Notes:} Results are based on 20{,}000 posterior draws.  “Lower” and “Upper” denote the $0.05$ and $0.95$ posterior quantiles, respectively.
\end{minipage}
\end{table}

For each $K$, we report the ETEL bootstrap posterior mean, standard deviation, median, and the 5th and 95th quantiles for both $\theta_0$ and $\theta_1$ in Table \ref{tab:posterior-summaries-K}. Although our projection-based posterior does not leverage the truncated Student-$t$ prior adopted in \cite{10.1111/rssb.12484}, its posterior means and credible intervals closely match the reported results across sieve dimensions. The pattern is also consistent with the observation in \cite{10.1111/rssb.12484}: increasing $K$ introduces more unconditional restrictions and therefore produces more concentrated posteriors. Across the sieve dimensions considered, the 90\% posterior credible intervals consistently contain the true parameter values, indicating that inference remains robust for various values of $K$.

\subsection{Asset Pricing} \label{subsec:ap}

We consider the asset-pricing application in \cite{10.1111/rssb.12484}, which uses monthly excess returns from January 1974 to December 2018 $(T=540)$. The data include $12$ candidate risk factors from the \texttt{czfactor} R package. Let $f_t=(x_t',w_t')'\in\mathbb{R}^{12}$, where $x_t$ denotes the market excess return and $w_t$ contains the remaining $11$ factors. The parameters of interest are $(\beta,\mu_x)$, where $\beta$ is the risk premium associated with the market factor and $\mu_x=\mathbb{E}(x_t)$. Following standard stochastic discount factor (SDF) theory, identification is based on a combination of unconditional pricing restrictions and conditional moment conditions,
\begin{align*}
\mathbb{E}\!\left[(1-\beta(x_t-\mu_x))f_t\right] = 0, \quad \mathbb{E}[x_t-\mu_x \mid f_{t-1}] &= 0.
\end{align*}

We convert the conditional restriction into unconditional moments as in \cite{10.1111/rssb.12484}:
\[
\mathbb{E}\!\left[(x_t-\mu_x) \otimes 
\bigl(q^K(f_{1,t-1}), \tilde q^K(f_{2,t-1}), \ldots, \tilde q^K(f_{12,t-1})\bigr)\right] = 0.
\]
Here $q^K(f_{1,t-1})$ consists of $K=3$ natural cubic spline basis functions as in the original analysis. For each factor $j \geq 2$,  $\tilde q^K(f_{j,t-1})$ contains two basis functions formed by subtracting the first column of $q^K(f_{j,t-1})$ from the remaining columns and then removing redundancy. This construction yields $3+(12-1)\cdot(3-1)=25$ expanded moment
conditions. Together with the $12$ pricing restrictions, this gives
$37$ moment conditions in total.

While \cite{10.1111/rssb.12484} employs a training-sample-based Student-$t$ prior on $(\beta,\mu_x)$, our projection-based bootstrap approach does not specify a direct prior on these parameters. Instead, posterior uncertainty is induced by reweighted empirical distributions
passed through the ETEL projection. To ensure comparability, we match their choices of $K$ and the number of posterior draws. Table \ref{tab:asset-pricing-compare} compares posterior summaries from our replication with those reported in \cite{10.1111/rssb.12484}. Although we do not specify a direct prior on $(\beta, \mu_x)$, our posterior estimates are close to their original results. Posterior uncertainty is modestly larger, reflecting the absence of prior regularization, but the $90\%$ credible interval continues to exclude zero by a wide margin. Consequently, both approaches lead to the same economic conclusion that the market excess return is a priced risk factor.

\begin{table}[t]
\centering
\caption{Posterior summaries for $(\beta,\mu_x)$ in the SDF.}
\label{tab:asset-pricing-compare}
\begin{tabular}{llrrrrr}
\toprule
Method & Parameter & Mean & SD & Median & Lower & Upper \\
\midrule
\multirow{2}{*}{ETEL bootstrap}
& $\beta$     & 2.823 & 0.802 & 2.918 & 1.367 & 4.249 \\
& $\mu_x$ & 0.006 & 0.002 & 0.006 & 0.003 & 0.009 \\
\multirow{2}{*}{BETEL with Student-$t$ prior}
& $\beta$     & 2.981 & 0.730 & 2.955 & 1.818 & 4.211 \\
& $\mu_x$ & 0.006 & 0.001 & 0.006 & 0.004 & 0.008 \\
\bottomrule
\end{tabular}

\vspace{0.5em}
\begin{minipage}{0.95\linewidth}
\footnotesize
\emph{Notes:} The row labeled ``BETEL with Student-$t$ prior'' reproduces the published posterior summaries reported in \cite{10.1111/rssb.12484}, with Student-$t$ prior on $(\beta, \mu_x)$. The results for both approaches are based on $50,000$ posterior draws.  ``Lower'' and ``Upper'' correspond to the $0.05$ and $0.95$ posterior quantiles.
\end{minipage}
\end{table}

\subsection{Average Treatment Effect (ATE)} \label{subsec:ate}
Consider the Massachusetts lottery data analyzed in \cite{10.1111/rssb.12294}, where the treatment indicator $w_i\in \{0,1\}$ denotes winning a large prize. Let $x_i \in \mathbb{R}^{13}$ be the covariate vector, $y_i$ be the average labor income over the six years following the lottery, and the propensity score $\eta_i = \Pr(w_i=1\mid x_i)=\exp(\gamma'x_i)/\{1+\exp(\gamma'x_i)\}$. To improve overlap, we trim observations with extreme propensity scores, following \cite{10.1111/rssb.12484}. This yields a filtered sample size of $N=323$. 

The parameter of interest is $\beta = (\gamma, \tau)$, where $\tau$ is the ATE. Define $z_i=(x_i,y_i,w_i)$. Estimation is based on the moment conditions
\[
\mathbb{E}\!\left[g(z_i,\beta)\right]=0, \qquad
g(z_i,\beta)=\left[\bigl\{x_i(w_i-\eta_i)\bigr\}^\prime,\;
\frac{(w_i-\eta_i)y_i}{\eta_i(1-\eta_i)}-\tau\right]^\prime.
\]

Our ETEL bootstrap procedure yields a posterior mean ATE of \(-\$5{,}997\) (SD \(1{,}516\)), with a 90\% posterior credible interval \([-\$8{,}434,\,-\$3{,}696]\). These results are close to those reported by \cite{10.1111/rssb.12294}, who obtain \(-\$5{,}346\) with interval \([-\$8{,}069,\,-\$2{,}720]\). Figure \ref{fig:tau} in the appendix shows that the posterior distribution is unimodal and closely aligned with the benchmark. Both approaches therefore lead to the same conclusion: lottery winnings are associated with a reduction in subsequent earnings. The similarity of the results indicates that the identifying information in the moment conditions is sufficient to recover the main effect, with prior assumptions primarily affecting posterior dispersion.

\section{AI-Powered ETEL} \label{sec:betel-ai}

The ETEL bootstrap in Section \ref{sec:betel-bootstrap} shows that the proposed projection posterior can deliver informative inference using only the observed sample and the moment restrictions. When additional regularization is desirable, our framework can incorporate auxiliary information through the prior base measure. In this section, we construct this base measure using large language models (LLMs), which can encode domain knowledge that may be difficult to express through a direct prior on $\theta$. Although our applications focus on LLMs, our framework applies more broadly to other generative models and auxiliary data sources.

Following \cite{ohagan2025aipoweredbayesianinference}, we generate synthetic data conditionally on observed covariates. For observed data $\mathcal{D}_n=\{x_i:=(z_i,y_i)\}_{i=1}^n$, where $z_i$ denotes covariates and $y_i$ denotes labels, we proceed draw by draw. For each posterior draw $b$, we first sample covariates $\{z_j^{(b)}\}_{j=1}^m$ from the empirical covariate distribution and then query the LLM to generate the corresponding synthetic labels $\{y_j^{*(b)}\}_{j=1}^m$. The augmented sample in draw $b$ is therefore $\{x_j^{*(b)}:=(z_j^{(b)},y_j^{*(b)})\}_{j=1}^m$.  Although one could instead use an LLM to simulate the entire data-generating process unconditionally, our empirical evidence suggests this approach may perform worse, as illustrated in Appendix \ref{subsec:full-dgp-generation}.

\subsection{Equity Return Predictions} \label{subsec:returns}

We study a firm-date level equity return prediction problem based on overnight news headlines. For firm $i$ on date $t$, define 
\vspace{-0.2cm}
$$y_{it} := \mathbbm{1}\left \{P_{i,t}^{\rm{open}} - P_{i,t-1}^{\rm{close}} >0 \right\},$$
the sign of the overnight return. Each observation consists of a bundle of headlines associated with firm $i$ between the previous market close and the next market open, together with the corresponding binary return label. We focus on the top $40$ U.S. firms by market capitalization as of June 30, 2025. The sample spans July 1 to December 31, 2025. We use July--August as the training period, with 1,253 firm-date
observations and 13,036 headlines, and September--December as the
evaluation period, with 2,582 firm-date observations and 25,721
headlines. For each firm-date, we concatenate the associated headlines into a single document and represent it by a $10,000$-dimensional sparse TF-IDF (term frequency–inverse document frequency) vector, using unigram and bigram features \cite{SALTON1988513}. 

Let $v_{it}$ denote the resulting TF-IDF feature vector, and let $\eta(\cdot)$ be the sigmoid function. The logistic score moment condition is $\mathbb{E}\!\left[v_{it}\bigl(y_{it}-\eta(v_{it}^{\top}\beta)\bigr)\right]=0$. In this exactly identified setting, GPT-ETEL can be viewed as a
Bayesian-bootstrap reweighting of the logistic score equations, augmented
by synthetic labels. Concretely, for each posterior draw we set $m=3000$, resample 3,000 firm-date observations from the empirical training distribution, prompt ChatGPT-5.2 using the corresponding raw headline bundles, and convert the returned sentiment scores into synthetic binary labels. Because the model cutoff date for ChatGPT-5.2 is August 31, 2025, and all synthetic labels are generated using only the training sample, this procedure does not introduce look-ahead bias when evaluating September-December data.

We evaluate predictive performance using $500$ Monte Carlo replications. In each replication, the September–December observations are randomly split into equal-sized validation and test sets. For the GPT-ETEL procedure, we select the prior strength parameter $\alpha \in \{0, 1, 10, 100, 200, 350, 500, 750, 1000\}$ by maximizing validation AUC (area under the ROC curve), computed on the validation split. As a baseline, we consider an $\ell_2$-logistic regression, and select the inverse penalty strength $C \in \{5\times 10^{-4}, 10^{-3}, 10^{-2}, 0.05, 0.1, 0.5, 1, 10, 100\}$ using the same criterion. After selecting the tuning parameter, we refit each model on the combined training and validation samples and report performance on the held-out
test set.

\begin{table}[!t]
\centering
\caption{Test-set performance in overnight news prediction.}
\label{tab:finance_main_results}
\small
\setlength{\tabcolsep}{7pt}
\renewcommand{\arraystretch}{1.15}
\begin{tabular}{lcccccc}
\toprule
& \multicolumn{3}{c}{AUC} & \multicolumn{3}{c}{Accuracy} \\
\cmidrule(lr){2-4} \cmidrule(lr){5-7}
Method & Mean & SD & Win rate & Mean & SD & Win rate \\
\midrule
GPT-ETEL
& 0.5743 & 0.0114 & 85.2\%
& 0.5605 & 0.0110 & 79.4\% \\

$\ell_2$-logistic
& 0.5597 & 0.0163 & --
& 0.5476 & 0.0143 & -- \\
\bottomrule
\end{tabular}

\vspace{0.5em}
\begin{minipage}{0.95\linewidth}
\footnotesize
\emph{Notes:} Results are based on $500$ Monte Carlo replications. Win rates report the proportion of Monte Carlo simulations in which GPT-ETEL outperforms the $\ell_2$-logistic benchmark.
\end{minipage}
\end{table}

Table \ref{tab:finance_main_results} shows that GPT-ETEL delivers a clear and consistent improvement over the $\ell_2$-logistic regression benchmark. Averaged over $500$ Monte Carlo replications, GPT-ETEL achieves higher mean test AUC and accuracy, as well as higher pairwise win rates. Figure \ref{fig:val_auc_alpha_l2} reports validation AUC across tuning parameters and shows that the candidate grids contain the empirical performance peaks for both methods. In particular, $\alpha=200$ appears to be the most favorable choice, indicating that the AI prior and synthetic labels can be beneficial in predicting returns. The test-set AUC is higher than the corresponding validation AUC because, after tuning, each method is re-estimated on the enlarged sample formed by combining the training and validation observations before being evaluated on the held-out test set.

\begin{figure}[!t]
\centering
\includegraphics[width=0.99\textwidth]{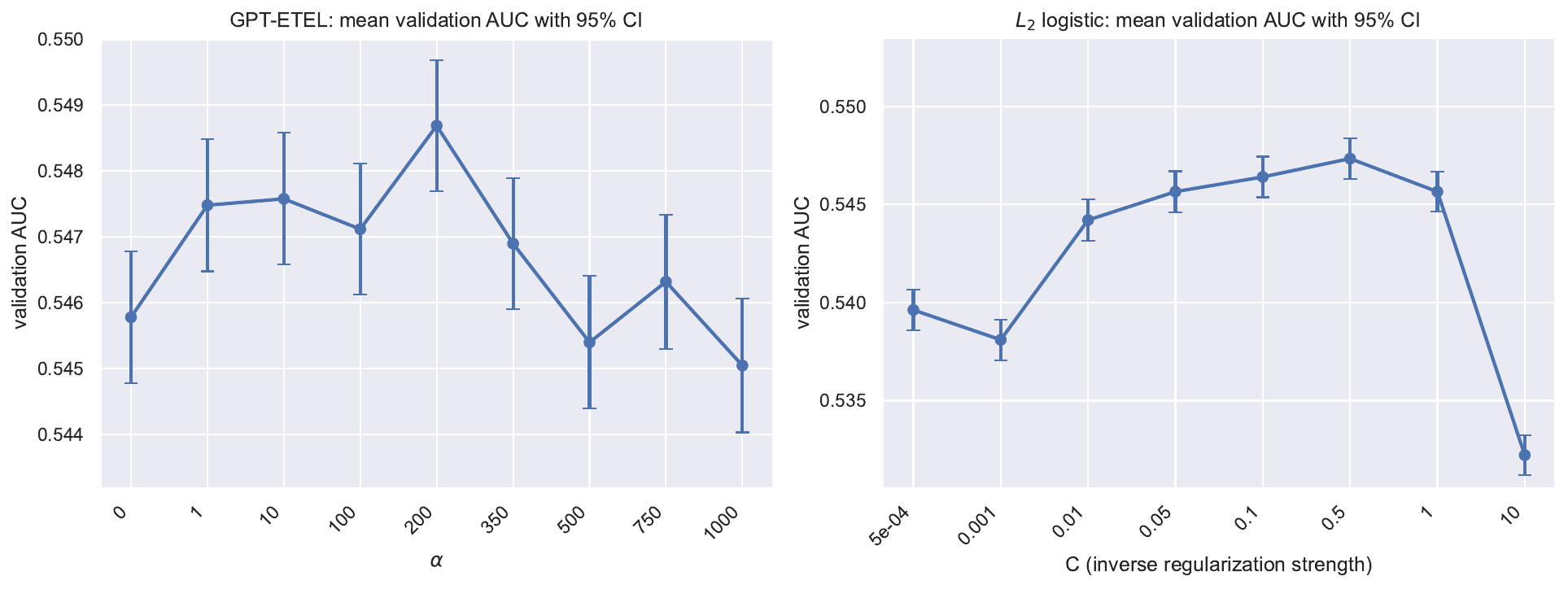}
\caption{Validation AUC across tuning parameters. Both panels report mean validation AUC with 95\% confidence intervals based on 500 Monte Carlo replications.}
\label{fig:val_auc_alpha_l2}
\end{figure}

More broadly, this experiment illustrates a key advantage of GPT-ETEL in prediction problems with unstructured text. In such settings, specifying a meaningful prior is difficult: a conventional Gaussian prior on the regression coefficients is convenient, but it does not meaningfully exploit the semantic content of the text itself. As a result, it often functions mainly as generic shrinkage and can appear somewhat arbitrary from a modeling standpoint. By contrast, our approach leverages an AI prior to extract information from the headline bundles, allowing the prior to be informed by textual content rather than by an abstract penalty on coefficients. In addition to improving predictive performance, this approach can reduce reliance on hand-crafted text-specific modeling choices. Additional details on data collection and alternative prompting strategies, including generating synthetic news directly for this prediction problem, are provided in Appendix \ref{sec:finance-appendix}.

\subsection{Engel Curve Recovery} \label{subsec:Engel}

We study the recovery of a parametric Engel curve calibrated to the application  in \cite{bck}. The data come from the 1995 British Family Expenditure Survey at the household level, where $y_i$ is food budget share, $x_i$ is log total expenditure, and $z_i$ is log gross earnings. Following \cite{bck}, we focus on $628$ working-age couples without children. The DGP is given by 
\begin{equation} \label{eq:Engel-dgp}
    y_i = h_0(x_i) + \varepsilon_i, \quad \varepsilon_i = \mathbb{E}[h_0(x_i) |z_i] - h_0(x_i) + v_i,
\end{equation}
where $h_0(x_i)$ is a nonlinear decreasing Engel curve and $v_i \sim \mathcal{N}(0, 0.01)$. Our goal is to recover $h_0(x_i)$ under the conditional moment restriction 
$\mathbb{E}[\varepsilon_i\mid z_i]=0$. The structural function is assumed to have a decreasing probit form,
\begin{equation*} \label{eq:Engel-curve}
    h_0(x) = a -b \Phi\left (\frac{x-c}{d} \right),
\end{equation*}
where $(a,b,c,d) \approx (0.28, 0.20, 5.34, 0.51)$ is estimated from the data. Since $b >0$, it is consistent with Engel's law that food budget share tends to fall as total expenditure rises. We estimate the joint distribution of $(x_i,z_i)$ from the entire sample using kernel methods, and denote the estimated density as $\hat{f}(x,z)$.

We report results from $100$ Monte Carlo simulations. In each simulation, we draw a training sample of size $N$ from $\hat f(x,z)$, generate outcomes from \eqref{eq:Engel-dgp}, and then generate an independent test sample of the same size. To construct the AI prior, for each posterior draw we sample $m=n/2$ observations of $(x_i,z_i)$ from the training data and use the OpenAI API to generate synthetic outcomes $y_i^*$ conditional on $(x_i,z_i)$. The prompt is designed to include Engel's law and to
encode qualitative shape information: the Engel curve should exhibit the
economically expected decreasing shape. This type of interpretable shape restriction is difficult to impose through a direct prior on the spline coefficients or through a generic frequentist regularization scheme. Additional prompt details are given in Appendix \ref{subsec:prompt-design}. 

For GPT-ETEL, we use $100$ posterior draws for various choices of $\alpha$, including the $\alpha=0$ case corresponding to the ETEL bootstrap without synthetic data. As a benchmark, we report the NPIV estimator of \cite{b9a32a52-d8c7-3448-8a03-b6226cc84dcd}, which regularizes the ill-posed inverse problem via a compactness restriction. For both GPT-ETEL and NPIV, we approximate the structural function using
the same cubic spline basis with $3$ degrees of freedom for $x$ and a richer spline basis with $4$ degrees of freedom for $z$.

Table \ref{tab:Engel-curve-rmse} reports the prediction RMSE for structural recovery. When $N=100$, the ETEL bootstrap performs noticeably worse than the NPIV benchmark. The absence of prior regularization exacerbates the convex-hull problem: the moment restrictions are enforced on a sparse support, leading to highly concentrated implied weights. Once synthetic observations are introduced through the AI prior, RMSE improves substantially. To illustrate, the left panel of Figure \ref{fig:cum-mass-sidebyside} plots the cumulative top-$k$ implied masses $\sum_{j=1}^k p_{(j)}$ for a representative Monte Carlo replication, where $p_{(j)}$ denotes the $j$-th largest ETEL-implied probability mass. Without AI augmentation, the two largest support points can constitute $10\%$ of the total weights. In contrast, GPT-ETEL substantially flattens the curve, especially for larger $\alpha$, indicating a more diffuse weight distribution. The right panel shows the corresponding recovered Engel curves. Relative to ETEL bootstrap, GPT-ETEL yields a more regularized curve that tracks the true shape more closely.

\begin{table}[t]
\centering
\caption{Training and test RMSE for the Engel-curve recovery.}
\label{tab:Engel-curve-rmse}
\begin{tabular}{lcccc}
\toprule
& \multicolumn{2}{c}{$N=100$} & \multicolumn{2}{c}{$N=400$} \\
\cmidrule(lr){2-3}\cmidrule(lr){4-5}
Method & Train RMSE & Test RMSE & Train RMSE & Test RMSE \\
\midrule
NPIV        & 0.02810 & 0.03127 & 0.01620 & 0.01717 \\
ETEL bootstrap ($\alpha=0$)         & 0.03406 & 0.03676 & 0.01609 & 0.01693 \\
GPT-ETEL ($\alpha=0.01N$) & 0.02875 & 0.03269 & 0.01608 & 0.01693 \\
GPT-ETEL ($\alpha=0.05N$) & 0.02867 & 0.03196 & \textbf{0.01607} & \textbf{0.01690} \\
GPT-ETEL ($\alpha=0.10N$) & 0.02510 & 0.02712 & 0.01610 & 0.01692 \\
GPT-ETEL ($\alpha=0.20N$) & 0.02618 & 0.02897 & 0.01617 & 0.01696 \\
GPT-ETEL ($\alpha=0.30N$) & 0.02630 & 0.02855 & 0.01632 & 0.01709 \\
GPT-ETEL ($\alpha=0.50N$) & \textbf{0.02429} & \textbf{0.02668} & 0.01663 & 0.01733 \\
\bottomrule
\end{tabular}

\vspace{0.4em}
\begin{minipage}{0.92\linewidth}
\footnotesize
\emph{Notes:} Boldface indicates the lowest RMSE within each column.
\end{minipage}
\end{table}

\begin{figure}[t]
\centering

\begin{minipage}[t]{0.49\textwidth}
    \centering
    \includegraphics[width=\textwidth]{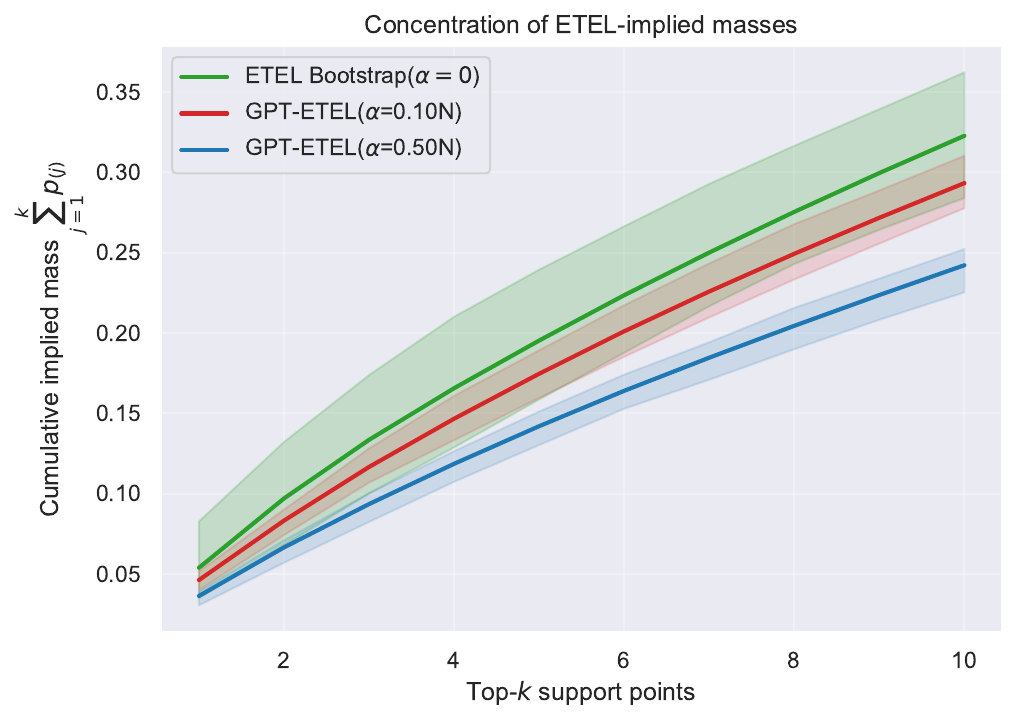}
\end{minipage}
\hfill
\begin{minipage}[t]{0.49\textwidth}
    \centering
    \includegraphics[width=\textwidth]{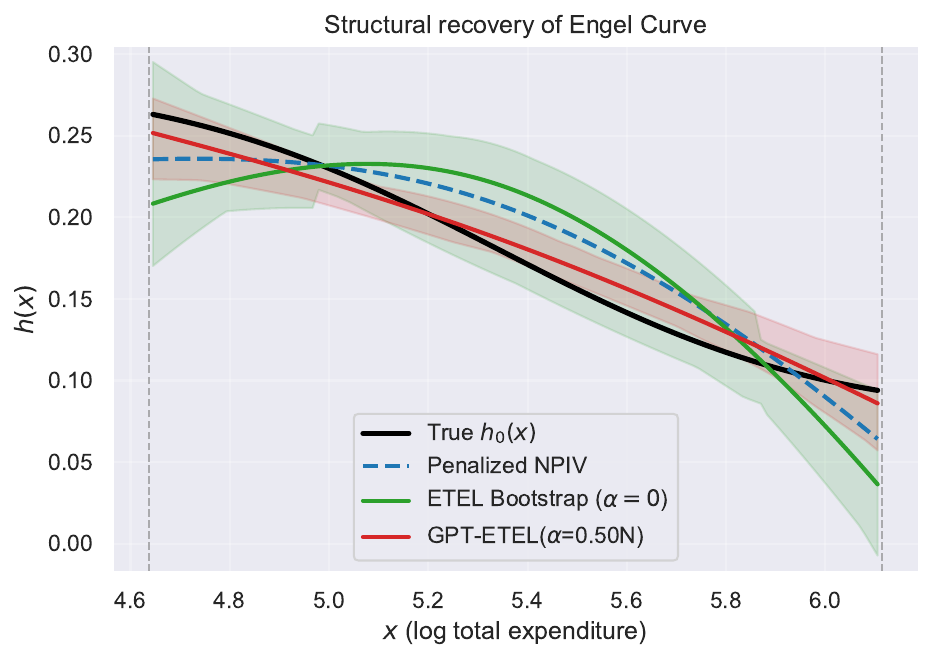}
\end{minipage}

\caption{Cumulative implied mass and structural function recovery ($N=100$). The shaded confidence bands are constructed from $100$ posterior draws.}
\label{fig:cum-mass-sidebyside}
\end{figure}

When $N=400$, the convex-hull issue is less severe. In this regime, ETEL bootstrap can already slightly outperform NPIV, and adding AI-generated prior information yields modest additional gains. More broadly, this experiment highlights the potential of AI-based prior in settings where directly specifying a prior on the parametric function may be difficult to justify. The GPT-based prior uses contextual economic information to generate economically grounded
pseudo-samples that incorporate additional shape information, which in turn improves inference.

\vspace{-0.3cm}
\section{Discussion}

We have proposed a nonparametric Bayesian approach to inference in moment condition models. Our approach places a Dirichlet process prior on the distribution of the observables, allowing for settings
in which prior information can be incorporated through auxiliary data or generative AI. This perspective is attractive when sample sizes are small, and prior domain knowledge is more naturally expressed at the distributional level. The resulting procedure is simple, parallelizable, and applicable both with and without auxiliary information. We established posterior consistency and Bernstein--von Mises theorems under regimes in which prior influence is either asymptotically negligible or asymptotically persistent. The empirical results show that the method can deliver reliable inference from moment restrictions alone, and can also leverage generative AI as a source of auxiliary information when direct priors on the parameter are difficult to specify.

There are several natural directions for future work. First, further theory is
needed to understand the posterior on $\theta$ induced by nonparametric priors
on $F$. The generalized ridge calculation in Section~\ref{subsec:reg} shows that,
in the linear regression setting, a central posterior summary admits an explicit
shrinkage representation. More explicit characterizations of the induced posterior would be useful for
extending the approach to richer semiparametric and nonparametric settings,
including models defined by general conditional moment restrictions
\citep{kankanala2025generalizedbayesconditionalmoment,liao2011posterior}. Second, AI-generated prior information should be made more reliable through calibration or rectification of the AI-induced base measure. Recent work on rectified AI priors shows that adjusting the synthetic data-generating law before using it can reduce the centering bias of AI-informed posteriors while preserving efficiency gains \cite{ChoiOHagan2026}. Developing such bias-reduction tools for moment restriction models is a promising route toward more robust AI-powered inference.

\section*{Data availability statement and acknowledgment}

Data used in Section \ref{subsec:ap} are available through the \texttt{czfactor} R package. The data used in Sections \ref{subsec:ate} and \ref{subsec:Engel} are public and are included with the replication code. The financial news data used in Section \ref{subsec:returns} are proprietary and accessible through a WRDS account; we provide code to query the relevant headlines and stock price data through the WRDS API.

The authors acknowledge the use of ChatGPT (OpenAI, GPT-5.5) to assist with language editing, code clarity, and the generation of synthetic data used in Section \ref{sec:betel-ai}.

\spacingset{1.0}
\bibliographystyle{chicago}
\bibliography{ell_final}
\clearpage
\newpage{}
\appendix
\begin{center}
    \section*{SUPPLEMENTARY MATERIALS}
\end{center}
\spacingset{1.2}

\section{Moment Condition Models: Motivating Examples} \label{sec:motivating-examples}

We provide common examples of statistical models that are naturally
formulated through moment conditions.

\begin{example}[Loss functions]\label{ex1}
Let $\ell(x,\theta)$ denote a loss function, and suppose that
$\theta_0 = \arg\min_{\theta \in \Theta} \mathbb{E}[\ell(x,\theta)]$,
where $\Theta \subseteq \mathbb{R}^{d_{\theta}}$. If $\ell(x,\theta)$ is differentiable in $\theta$ and $\theta_0$ lies in the interior of $\Theta$, then the first-order condition implies
\[
\mathbb{E}\!\left[\nabla_\theta \ell(x,\theta_0)\right] = 0.
\]
In this case, the moment function is  $g(x,\theta) = \nabla_\theta \ell(x,\theta)$.
\end{example}

\begin{example}[Linear instrumental variables]\label{ex2}
Let $x = (y,D,Z^\prime)^\prime$, where $y$ is the outcome variable, $D$ is the regressor of interest, and $Z$ is a vector of instruments. Consider the linear model $
y = D\theta_0 + u $, where $u$ is an unobserved error term and $D$ may be endogenous, meaning that it need not satisfy the usual exogeneity condition $\mathbb{E}[Du] = 0$. Identification relies on a vector of instruments $Z$ that is correlated with $D$ but orthogonal to $u$, so that 
\[
\mathbb{E}\!\left[Z\bigl(y - D\theta_0\bigr)\right] = 0.
\]
\end{example}

\begin{example}[Quantile regression]\label{ex3}
Let $x = (y,z)$, where $y$ is the outcome and $z$ is a vector of covariates, and fix a quantile $\tau \in (0,1)$. In the linear quantile regression model \cite{koenker1978regression}, the parameter $\theta_0$ characterizes the $\tau$-th conditional quantile of $y$ given $z$ through $
Q_\tau(y \mid z) = z^\prime \theta_0.$
Under the usual regularity conditions, $\theta_0$ satisfies the moment restriction
\[
\mathbb{E}\!\left[z\left\{\tau - \mathbbm{1}(y \leq z^\prime \theta_0)\right\}\right] = 0.
\]
\end{example}

\begin{example}[Dynamic panels]\label{ex:dynpanel}
Let $
x_i=\{(y_{it},d_{it}) : t=0,\dots,T\}
$
denote the observed panel for unit $i$, where $y_{it}$ is the outcome and $d_{it}$ is a covariate of interest. Consider the dynamic panel model $
y_{it}=\alpha_0 y_{i,t-1}+\beta_0 d_{it}+\eta_i+u_{it}, $
where $\eta_i$ is an unobserved unit-specific effect and $u_{it}$ is an idiosyncratic error term. In this setting, $y_{i,t-1}$ may be correlated with $\eta_i$, and $d_{it}$ may be correlated with the contemporaneous shock $u_{it}$. A standard approach, following \cite{arellano1991some}, is to first-difference the model to eliminate $\eta_i$:
\[
\Delta y_{it}
=
\alpha_0 \Delta y_{i,t-1}
+
\beta_0 \Delta d_{it}
+
\Delta u_{it},
\qquad
\Delta u_{it}=u_{it}-u_{i,t-1}.
\]
If $u_{it}$ is serially uncorrelated, then lagged levels dated $t-2$ and earlier are orthogonal to $\Delta u_{it}$ and can therefore be used as instruments. For example, taking $
z_{it}=(y_{i,t-2},d_{i,t-2})^\prime$ for $t \geq 2$, the parameter vector $\theta_0=(\alpha_0,\beta_0)^\prime$ satisfies the moment restriction
\[
\mathbb{E}\!\left[
z_{it}
\Bigl(
\Delta y_{it}
-
\alpha_0 \Delta y_{i,t-1}
-
\beta_0 \Delta d_{it}
\Bigr)
\right]
=0.
\]
\end{example}

\section{Connection to Bayesian Bootstrap}

When $\alpha =0$, the ETEL projection posterior is supported only on the observed data, which aligns with the idea of Bayesian bootstrap \cite{10.1214/aos/1176345338}. Pinpointing the exact prior under Bayesian bootstrap is generally challenging, unless in very simple scenarios.  Appendix \ref{subsec:bb-example} demonstrates a simple example where we can explicitly deduce that Bayesian bootstrap induces a Haldane prior on $\theta$. In the weighted likelihood bootstrap (WLB) literature \cite{nthesis}, WLB draws can asymptotically match the parametric Bayesian posteriors in a higher order if the squared Jeffery prior is placed on $\theta$.

\subsection{Example: Prior Derivation Under Bayesian Bootstrap} \label{subsec:bb-example}

\begin{example}
Consider a simple case $\alpha=m=0$. Suppose the data $x \in \{0,1\}$ is binary, and there are $k$ ones and $(n-k)$ zeros. We are interested in the mean, so the moment condition is $g(x,\theta)= x-\theta$. With Bayesian bootstrap and let $v \sim\text{Dir}(1,\ldots,1)$ be the Dirichlet draw, we have the estimator $$\theta^*(V_n) = \sum_{i=1}^n v_i x_i = \sum_{i: x_i=1} v_i \sim \text{Beta}(k, n-k).$$
Given the likelihood is $\theta^{k}(1-\theta)^{n-k}$, and the posterior proportional to $\theta^{k-1}(1-\theta)^{n-k-1}$. Then we can deduce the prior is $\theta^{-1}(1-\theta)^{-1}$, which corresponds to the Haldane prior. In this particular example, our DGP prior is improper and the induced Haldane prior remains improper.
\end{example}

\subsection{Proof of Theorem \ref{thm:bootstrap}: Connection to Chamberlain and Imbens' Bayesian Bootstrap} \label{subsec:e-proof}

\begin{proof}
    When $\alpha=0$, there is no augmented data so that $m=0$.  write $F_v=\sum_{i=1}^n v_i\delta_{x_i}.$ For a fixed $\theta$, let $p_i(\theta)=v_iw_i^*(\theta)$ denote the projected probability mass. It is convenient to rewrite the ETEL criterion as follows:
    $$L(\theta):= - \sum_{i=1}^n v_i \log w_i^*(\theta) = - \sum_{i=1}^n v_i \log \frac{p_i(\theta)}{v_i} = -\sum_{i=1}^n v_i \log p_i(\theta) + \sum_{i=1}^n v_i \log v_i.$$
    
    Since the second term of $L(\theta)$ is constant in $\theta$, we may focus only on the first term which is the classic cross-entropy loss. By Gibbs' inequality, we have
    $$-\sum_{i=1}^n v_i\log p_i \ge -\sum_{i=1}^n v_i\log v_i,$$
    with equality if and only if $p_i=v_i$ for all $i$. By assumption, there exists a unique $\theta^\dagger$ such that $\sum_{i=1}^n v_i g(x_i,\theta^\dagger)=0.$ It follows that, at $\theta^\dagger$, the choice $p_i=v_i$, equivalently $w_i^*(\theta^\dagger)=1$, is feasible for the weighted entropy program. It attains the lower bound above, so $L(\theta^\dagger)=0$ and $\theta^\dagger$ minimizes the ETEL criterion. On the other hand, since $\theta^*(v)$ minimizes the ETEL criterion which is nonnegative (since it is also a KL criterion), we have $\theta^*(v)=\theta^\dagger.$ Since $\theta^\dagger$ is also the unique solution to the weighted moment equation defining $\hat\theta_{\rm CI}(v)$, we have
    $\hat\theta_{\rm CI}(v)=\theta^\dagger=\theta^*(v).$ Moreover, $w_i^*(\theta^\dagger)=1$ implies $P_{\theta^\dagger}^*(F_v)=F_v$ and $\ell(\theta^\dagger;F_v)=0$.
\end{proof}

\section{Proofs for Asymptotic Normality} \label{sec:bvm-proofs}

We provide the proofs of the theoretical results stated in Section \ref{sec:theory}. 
Section \ref{subsec:notation} reviews the notation used throughout the appendix.  We then establish posterior consistency in two steps: Section \ref{subsec:al-consistency} collects the auxiliary lemmas, and Section \ref{subsec:consistency-proof} proves Theorem \ref{thm:posterior-consistency}. Next, Sections \ref{subsec:al-bootstrap} and \ref{subsec:proof-bootstrap} present the auxiliary lemmas and proof of the BvM theorem under an asymptotically negligible prior. 
Finally, Section \ref{subsec:proof-non-vanishing} gives the proof for the non-vanishing prior-strength regime.

\subsection{Notations setup} \label{subsec:notation}
Our theory focuses on a continuous $F_{\rm{AI}}$ base measure. Recall that in Algorithm \ref{alg:ai-betel}, we approximate the DP posterior draw as
$$F^{(b)}= \sum_{i=1}^n v_i^{(b)}\,\delta_{x_i}+ \sum_{j=1}^m v_{n+j}^{*(b)}\,\delta_{x_j^{*(b)}},$$
where $x_1^{*(b)},\ldots,x_m^{*(b)} \stackrel{\mathrm{iid}}{\sim} F_{\rm{AI}}$ and $(v_1^{(b)}, \ldots, v_n^{(b)}, v_{n+1}^{*(b)}\ldots,v_{n+m}^{*(b)}) \sim \mathrm{Dirichlet}\!\left( \underbrace{1,\ldots,1}_{n}, \underbrace{\alpha/m,\ldots,\alpha/m}_{m}
    \right)$.
When $\alpha_n =0$, we set $m=0$.

For notational simplicity, we suppress the dependence on the draw $b$ and relabel the synthetic data to write 
\begin{align*}
    &V_n := (v_1^{(b)}, \ldots, v_n^{(b)}, v_{n+1}^{*(b)}\ldots,v_{n+m}^{*(b)})  = (v_1, \ldots, v_{n+m}), \\
    &(x_1^*, \ldots, x_m^{*}):= (x_{n+1}, \ldots, x_{n+m}).
\end{align*}
Our generic posterior draw is defined as 
\begin{equation*}
    \theta^{*}(V_n) := \arg \min_{\theta} l_{n,m,V}(\theta), \quad l_{n,m,V}(\theta) := - \sum_{k=1}^{n+m} v_k \log w_k^*(\theta),
\end{equation*}
where $w_k^{*}(\theta)$ is defined in (\ref{eq:AI-SOL}). Plugging the expression of $w_k^{*}(\theta)$ into $l_{n,m,V}$ yields
\begin{align*}
    l_{n,m,V}(\theta) = \psi_{n,m,V}(\lambda(\theta), \theta) - \lambda(\theta)^{\top} S_{n,m,V}(\theta),
\end{align*}
where 
\begin{equation*} \label{eq:loss-component}
\begin{aligned}
  S_{n,m,V}(\theta) 
  &:= \sum_{k=1}^{n+m} v_k g(x_k, \theta), \qquad \psi_{n,m,V}(\eta, \theta)
  := \log \sum_{k=1}^{n+m} v_k e^{\eta^{\top} g(x_k, \theta)}, \\
  \lambda(\theta)
  &:= \arg\min_{\eta\in\Lambda} \psi_{n,m,V}(\eta,\theta).
\end{aligned}
\end{equation*}

Finally, we define their population counterparts as
\begin{equation*} \label{eq:population loss-component}
\begin{aligned}
  & \bar g(\theta) := \mathbb E_{F_0}[g(X,\theta)],\qquad \psi_0(\eta,\theta) := \log \mathbb E_{F_0} \left[\exp\{\eta^\top g(X,\theta)\} \right],
  &  \lambda_0(\theta) := \arg\min_{\eta\in\Lambda} \psi_0(\eta,\theta), \\
  &  L_0(\theta) :=  \psi_0\{\lambda_0(\theta),\theta\}
    - \lambda_0(\theta)^\top\bar g(\theta) =
    D_{\mathrm{KL}} \left( F_0\,\|\,P_\theta^*(F_0) \right).
\end{aligned}
\end{equation*}

\subsection{Auxiliary Lemmas for Theorem \ref{thm:posterior-consistency}}
\label{subsec:al-consistency}
\begin{lemma}[Dirichlet Weighted LLN]\label{lem:wlln}
Suppose the following global envelope conditions in Theorem \ref{thm:posterior-consistency} hold
$$\mathbb E_{F} \sup_{\theta\in\Theta}\|g(X,\theta)\| <\infty,
\qquad \mathbb E_{F} \sup_{\theta\in\Theta,\eta\in\Lambda} e^{\eta^\top g(X,\theta)}
<\infty,
$$
for $F=F_0$ and, when $\alpha_n>0$, for $F=F_{\rm AI}$.
If $\alpha_n/n\to0$, then for the classes $\mathcal{H}_1:=\{g(\cdot,\theta):\theta\in\Theta\}$ and $\mathcal{H}_2=\{\exp\{\eta^\top g(\cdot,\theta)\}:\eta\in \Lambda,\ \theta\in \Theta\}$,
$$
\sup_{\theta\in \Theta}\Big\| S_{n,m,V}(\theta)-\bar g(\theta)\|=o_p(1),\quad
\sup_{\eta \in \Lambda,\ \theta\in\Theta}\Big| \psi_{n,m,V}(\eta,\theta)-\psi_0(\eta,\theta)\Big|=o_p(1).
$$

Additionally, under the full Assumptions \ref{A1}-\ref{A4}, the same weighted LLN statements hold, locally over $\mathcal N$, for the derivative and product
classes needed below:  $\mathcal{H}_3 = \{\nabla_{\theta}g(., \theta): \theta \in \mathcal{N}\}$, $\mathcal{H}_4 = \{g(., \theta)g(., \theta)^{\top}: \theta \in \mathcal{N}\}$, $\mathcal{H}_5 = \{\nabla_{\theta}^2 g(., \theta): \theta \in \mathcal{N}\}$, $\mathcal{H}_6 = \{\|g\| \|\nabla_{\theta}g(., \theta)\|: \theta \in \mathcal{N}\}$,  $\mathcal{H}_7 = \{e^{\eta^{\top} g(\cdot, \theta)}\|g(\cdot, \theta)\|^3: \eta \in \Lambda,  \theta \in \mathcal{N}\}$.
\end{lemma}

\begin{proof}
    Recall that 
      $$v_1, \ldots, v_n, v_1^*, \ldots, v_m^{*} \sim \text{Dirichlet}(1, \ldots, 1, \frac{\alpha_n}{m}, \ldots , \frac{\alpha_n}{m}).$$
    It is more convenient to work with the gamma representation of the Dirichlet weights. Let $\gamma_i \overset{\text{i.i.d.}}{\sim} \Gamma(1,1)$ (equivalent to exponential weights) for $i\in \{1, \ldots, n\}$ and $\gamma_{n+j} \overset{\text{i.i.d.}}{\sim} \Gamma(\alpha_n/m,1)$ for $j\in \{1, \ldots, m\}$, then we have $v_k := \frac{\gamma_k}{\sum_{i=1}^{n+m}\gamma_i}$.  The Dirichlet Weighted LLN statement is a consequence of Lemma 3 of Chapter 3 in Newton's dissertation \cite{nthesis}, which asserts that, for any integrable real-valued $h(.)$, $\frac{1}{n}\sum_i Y_i h_i(.)$ converges to $\mathbb{E}[h(.)]$ in probability for exponential weights $Y_i \overset{\text{i.i.d.}}{\sim} \text{Exp}(1)$. 
    
    We prove the Lemma by considering $g(x, \theta) \in \mathcal{H}_1$, and the same argument still holds for other integrable function classes $\mathcal{H}_2-\mathcal{H}_7$ with $g(\cdot, \theta)$ replaced by the corresponding functional forms. Additionally, for the global classes $\mathcal H_1$ and $\mathcal H_2$, the preceding argument is applied over $\Theta$. For the classes $\mathcal H_3$--$\mathcal H_7$, it is applied only locally over $\mathcal N$.
    
    To this end, we write $W_n:=\sum_{i=1}^n \gamma_i$, $W_{\alpha_n}:= \sum_{j=1}^m \gamma_{n+j}$, and $\delta_n:= \frac{W_{\alpha_n}}{W_n + W_{\alpha_n}}$. Recall that in (\ref{eq:loss-component}) we have defined $S_{n,m,V}(\theta) := \sum_{k=1}^{n+m} v_k g(x_k, \theta)$. We can reparametrize it as follows:
    $$S_{n,m,V}(\theta):= (1-\delta_n) \sum_{i=1}^n \omega_i g(x_i, \theta)+ \delta_n \sum_{j=1}^m \pi_j g(x_j^*, \theta),$$
    where $\omega_i := \frac{\gamma_i}{W_n}$ and $\pi_j := \frac{\gamma_{n+j}}{W_{\alpha_n}}$. Since $\frac{\alpha_n}{n} \rightarrow 0$, it follows $\delta_n$ is asymptotically negligible since $\delta_n = O_p(\frac{\alpha_n}{n}) = o_p(1)$. It is then sufficient to show
    $$\sup_{\theta \in \mathcal{N}} \|\sum_{i=1}^n w_i g(x_i, \theta)-\mathbb{E}[g(x,\theta)]\|=o_p(1),\quad \sup_{\theta \in \mathcal{N}}\|\sum_{j=1}^{m} \pi_j g(x_j^*,\theta)\|= O_p(1).$$
    By Lemma 3 in \cite{nthesis}, we have $\sum_{i=1}^n \omega_i g(x_i, \theta) \xrightarrow{p}  \mathbb{E}[g(x,\theta)]$ for each fixed $\theta$. Since the function class $\mathcal{H}_1$ is $P$-Glivenko–Cantelli under assumptions \ref{A2}, we have established uniform convergence over the compact neighborhood $\mathcal{N}$: $\sup_{\theta \in \mathcal{N}} \|\sum_{i=1}^n \omega_i g(x_i, \theta)-\mathbb{E}[g(x,\theta)]\|=o_p(1)$. The synthetic part $\sup_{\theta \in \mathcal{N}}\|\sum_{j=1}^{m} \pi_j g(x_j^*,\theta)\|= O_p(1)$ because the relevant envelope is integrable by assumption \ref{A2}.

    The same argument still holds for the other integrable function classes $\mathcal{H}_2-\mathcal{H}_7$ with $g(\cdot, \theta)$ replaced by the corresponding envelopes.
\end{proof}

\begin{lemma}[Consistency of ETEL Loss]\label{lem:global-consistency} 
Let $L_0(\theta):=D_{\mathrm{KL}}\!\left(F_0\,\|\,P_\theta^*(F_0)\right)$ be the standard ETEL criterion. Suppose the conditions in Theorem \ref{thm:posterior-consistency} hold.
Then,
$$\sup_{\theta\in\Theta} \left| l_{n,m,V}(\theta)- L_0(\theta)\right|=o_p(1).$$
Moreover, $L_0$ is continuous and has the unique minimizer $\theta_0$.
\end{lemma}

\begin{proof}
    Recall that $L_0(\theta)$ can be rewritten as
    $$L_0(\theta) = \psi_0(\lambda_0(\theta),\theta) -\lambda_0(\theta)^\top\bar g(\theta),$$
    where $\bar g(\theta) := \mathbb{E}_{F_0}[g(x,\theta)]$ and $\psi_0(\eta, \theta):= \log \mathbb{E}_{F_0} \left [e^{\eta^{\top} g(x,\theta)}\right]$, as defined in Assumptions \ref{A1} and \ref{A3} respectively. By Lemma \ref{lem:wlln}, we have 
    $$\sup_{\theta\in\Theta} \|S_{n,m,V}(\theta)-\bar g(\theta)\|=o_p(1), \qquad \sup_{\theta\in\Theta,\ \eta\in\Lambda} \left| \psi_{n,m,V}(\eta,\theta)-\psi_0(\eta,\theta)\right|=o_p(1).$$
   By compactness of $\Theta$ and $\Lambda$, joint continuity of $\psi_0$, and pointwise uniqueness in Assumption \ref{A3}, the population dual problem is uniformly separated. Therefore, by the global uniform convergence of $\psi_{n,m,V}(\theta)$ to $\psi_0$ and the uniform argmin theorem,
    $$\sup_{\theta\in\Theta} \|\lambda(\theta)-\lambda_0(\theta)\|=o_p(1).$$ 
    Since $\Lambda$ is compact, both $\lambda_{n,m,V}(\theta)$ and $\lambda_0(\theta)$ are uniformly bounded. Therefore, we have
    \begin{align*}
    \sup_{\theta\in\Theta} \left| l_{n,m,V}(\theta)-L_0(\theta)
    \right|  &\leq \sup_{\eta\in\Lambda,\theta\in\Theta} 
    \left| \psi_{n,m,V}(\eta,\theta)-\psi_0(\eta,\theta) \right| +
    \sup_{\theta\in\Theta} \left| \psi_0\{\lambda_{n,m,V}(\theta),\theta\}
    - \psi_0\{\lambda_0(\theta),\theta\} \right| \\
    &\qquad+ \sup_{\theta\in\Theta} \left| \lambda_{n,m,V}(\theta)^\top
    \{S_{n,m,V}(\theta)-\bar g(\theta)\}\right| +
    \sup_{\theta\in\Theta} \left| \{\lambda_{n,m,V}(\theta)-\lambda_0(\theta)\}^\top \bar g(\theta) \right|.
\end{align*}
The first and the third terms are $o_p(1)$ by Lemma \ref{lem:wlln}. 
The second term is also $o_{p}(1)$ by uniform continuity of $\psi_0$ on
$\Lambda\times\Theta$ and the uniform convergence of $\lambda_{n,m,V}$ to $\lambda_0$. Finally, the fourth term is $o_{p}(1)$ because $\bar g$ is continuous on compact $\Theta$ and hence uniformly bounded. It follows that
$$ \sup_{\theta\in\Theta} \left| l_{n,m,V}(\theta)-L_0(\theta) \right|
    = o_{p}(1).$$

    Additionally, it is clear that $\theta_0$ is the unique minimizer of $L_0(\theta)$: since $L_0(\theta)$ is a KL criterion, the smallest value it can obtain is $0$, which implies $P_{\theta}^*(F_0) = F_0$. Since  $P_{\theta}^*(F_0)$ satisfies the moment restriction, this implies $\mathbb{E}_{F_0}[g(x, \theta)]=0$, and hence $\theta=\theta_0$ by Assumption \ref{A1}. 
\end{proof}

\subsection{Proof of Theorem \ref{thm:posterior-consistency}} \label{subsec:consistency-proof}

\begin{proof}
    Fix $\varepsilon>0$ and define 
    $$A_\varepsilon:=
    \{\theta\in\Theta:\|\theta-\theta_0\|\geq\varepsilon\}, \qquad c_\varepsilon := \inf_{\theta\in A_\varepsilon} \{L_0(\theta)-L_0(\theta_0)\}.$$
    Since $L_0(\theta)$ is continuous and has the unique minimizer $\theta_0$, we have $c_{\varepsilon} > 0$. Suppose the bad event  $\{\theta^*(V_n)\in A_\varepsilon\}$ holds, then we have 
    \begin{align*}
    c_\varepsilon &\leq L_0\{\theta^*(V_n)\}-L_0(\theta_0) \\
    &\leq
    \left[L_0\{\theta^*(V_n)\}-l_{n,m,V}\{\theta^*(V_n)\}\right]
    +\left[ l_{n,m,V}(\theta_0)-L_0(\theta_0) \right]  \\
    &\leq 2  \sup_{\theta\in\Theta}|l_{n,m,V}(\theta)-L_0(\theta)| = o_p(1),
\end{align*}
where the first inequality is due to the definition of $c_{\varepsilon}$, and the second inequality is due to the fact that  $\theta^*(V_n)$ minimizes $l_{n,m,V}$ over $\Theta$, and the last $o_p(1)$ statement is due to Lemma \ref{lem:global-consistency}.
It follows that
$$\Pi_n \left(\|\theta^*-\theta_0\|\geq\varepsilon\mid \mathcal{D}_n
    \right) \leq \mathbb{P}(2\sup_{\theta\in\Theta}|l_{n,m,V}(\theta)-L_0(\theta)|  \geq  c_\varepsilon | \mathcal{D}_n) \xrightarrow{\mathbb{P}} 0.$$
    
    \end{proof}

\subsection{Auxiliary Lemmas for Theorem \ref{thm:bootstrap}} \label{subsec:al-bootstrap}

\begin{lemma} \label{lem:error}
Under Assumptions \ref{A1}-\ref{A4}:
\begin{enumerate}
    \item The operator norm $\| \nabla_{\eta}^3 \psi(\bar{\eta}, \theta)\|$ is uniformly bounded for $\eta\in\Lambda$ and $\theta \in \mathcal{N}$.
    \item For $\theta \in \mathcal{N}$, $\|\lambda(\theta)\| \leq C_1 \|S(\theta)\|$ for some constant $C_1$.
\end{enumerate}
\end{lemma}
\begin{proof}
  Define the tilted weight $p_k(\eta, \theta):= \frac{v_k e^{\eta^{\top} g_k(\theta)}}{\sum_j v_j e^{\eta^{\top} g_j(\theta)}}$ and let $\mu(\eta, \theta):= \nabla_{\eta} \psi(\eta ,\theta) = \sum_k p_k(\eta,\theta) g_k(\theta)$. By viewing $\psi(\eta,\theta)$ as the log cumulant generating function of $g_k(\theta)$, we can upper bound the operator norm of the third derivative tensor as
  $$\| \nabla^3 \psi(\eta, \theta) \| = \sup_{\|u\| = \|v\| = \|w\| = 1} | \sum_k p_k \prod_{t \in \{u,v,w\}} ((g_k-\mu)^{\top} t) | \leq \sum_k p_k \|g_k -\mu\|^3.$$
  Using $\|a-b\|^3 \leq 4(\|a\|^3 + \|b\|^3)$, the quantity above can be further bounded by $$4 \sum_k p_k \|g_k\|^3 + 4 \| \mu \|^3 \leq 8 \sum_k p_k \|g_k\|^3,$$
  where we have applied Holder's inequality to upper bound $\|\mu\| = \|\sum_k p_k g_k\| \leq \sum_k p_k \|g_k\| \leq (\sum_k p_k \|g_k\|^3)^{1/3}$. It follows that
  $$\| \nabla^3 \psi(\eta, \theta) \| \leq 8 \frac{\sum_k v_k e^{\eta^\top g_k(\theta)}\|g_k(\theta)\|^3} {\sum_j v_j e^{\eta^\top g_j(\theta)}}.$$
  By Lemma \ref{lem:wlln} applied to $\mathcal{H}_2$ and $\mathcal{H}_7$, both $\sum_k v_k e^{\eta^{\top} g_k(\theta)}$ and $\sum_k v_k e^{\eta^{\top} g_k(\theta)} \|g_k(\theta)\|^3$ converge uniformly to their finite population expectations. Hence the operator norm $\| \nabla_{\eta}^3 \psi(\bar{\eta}, \theta)\|$ is uniformly bounded.

  To prove the second statement, by Lemma \ref{lem:wlln}, 
  $$\sup_{\eta\in\Lambda,\theta\in\mathcal N}
    |\psi_{n,m,V}(\eta,\theta)-\psi_0(\eta,\theta)|=o_p(1).$$
By compactness of $\Lambda$ and $\mathcal N$, continuity of $\psi_0$,
and uniqueness of $\lambda_0(\theta)$ in Assumption \ref{A3}, the uniform argmin theorem as illustrated in Lemma \ref{lem:global-consistency} also gives $$\sup_{\theta\in\mathcal N} \|\lambda(\theta)-\lambda_0(\theta)\|=o_p(1).$$
Given that $\lambda_0(\theta_0)=0$ and $\lambda_0(\theta)$ is continuous, we may shrink $\mathcal N$, without changing the previous assumptions, so that $\sup_{\theta\in\mathcal N}\|\lambda_0(\theta)\|$ is arbitrarily small. By the first-order condition that $\nabla_\eta\psi_{n,m,V}(\lambda(\theta),\theta)=0$, a Taylor expansion around $\eta=0$ gives:
$$ 0 = S(\theta) + \hat\Omega_{n,m,V}(\theta)\lambda(\theta)+r_\lambda(\theta),
\qquad \|r_\lambda(\theta)\| \leq
    C\|\lambda(\theta)\|^2.$$
By Assumption \ref{A1} and Lemma
\ref{lem:wlln}, there exists $\kappa>0$ such that, with probability tending to one, $\inf_{\theta\in\mathcal N} \lambda_{\min}\{\hat\Omega_{n,m,V}(\theta)\} \geq \kappa $. Taking the inner product of the Taylor expansion with $\lambda(\theta)$ yields
 $$0 = \lambda(\theta)^\top S(\theta) +
    \lambda(\theta)^\top \hat\Omega_{n,m,V}(\theta)\lambda(\theta)+\lambda(\theta)^\top r_\lambda(\theta),$$
which implies
$$\kappa\|\lambda(\theta)\|^2 \leq \|\lambda(\theta)\|\,\|S(\theta)\|
    + C\|\lambda(\theta)\|^3.$$
Since  $\sup_{\theta\in\mathcal N}\|\lambda(\theta)\|$ can be made
smaller than $\kappa/(2C)$ with probability tending to one, the last term can be upper bounded by $ C\|\lambda(\theta)\|^3 \leq \frac{\kappa}{2}\|\lambda(\theta)\|^2$ uniformly over $\mathcal{N}$. It follows that
$ \frac{\kappa}{2}\|\lambda(\theta)\|^2 \leq  \|\lambda(\theta)\|\,\|S(\theta)\|.$
If $\lambda(\theta)=0$, the second statement is trivial. Otherwise, dividing both sides by $\|\lambda(\theta)\|$ proves the result.
\end{proof}

\begin{lemma} \label{lem:loss-2nd-d}
Under Assumptions \ref{A1} - \ref{A4}, uniformly over $\mathcal{N}$:
$$ \nabla_{\theta}^2l_{n,m,V}(\theta) = G_{n,m,V}(\theta)^{\top} \hat{\Omega}_{n,m,V}^{-1}(\theta) G_{n,m,V}(\theta) + O_p(\|S(\theta)\|).$$
\end{lemma}
\begin{proof}
We have
 \begin{align*}
        \nabla_{\theta}^2l_{n,m,V}(\theta) &= G_{n,m,V}(\theta)^{\top} \hat{\Omega}_{n,m,V}^{-1} G_{n,m,V}(\theta) + \{(\nabla_{\theta}G_{n,m,V}(\theta)^{\top}) \hat{\Omega}_{n,m,V}^{-1} S(\theta)\\ 
        & \quad - G_{n,m,V}(\theta) \hat{\Omega}_{n,m,V}^{-1} (\nabla_{\theta} \hat{\Omega}_{n,m,V}) \hat{\Omega}_{n,m,V}^{-1} S(\theta) \} + O(\|S(\theta)\|). \\
\end{align*}
It is easy to see $\nabla_{\theta} G_{n,m,V}(\theta)$ is uniformly bounded over $\mathcal{N}$ by assumptions \ref{A2} and Lemma \ref{lem:wlln}:
\begin{align*}
    \sup_{\theta \in \mathcal{N}} \|\nabla_{\theta} G_{n,m,V}(\theta)\| &= \sup_{\theta \in \mathcal{N}} \|\mathbb{P}_{n,m,V} \nabla_{\theta}^2 g(., \theta) \|   \\
    & \leq \mathbb{P}_{n,m,V} \| \sup_{\theta \in \mathcal{N}} \nabla_{\theta}^2 g(., \theta) \| = \mathbb E_{F_0}[L_2(X)] + o_p(1) = O_p(1).
\end{align*}
To show that $\nabla_{\theta} \hat{\Omega}(\theta)$ is also $O_p(1)$, note that 
\begin{align*}
\nabla_\theta \hat\Omega_{n,m,V}(\theta)&=
    \mathbb P_{n,m,V} \left[ \nabla_\theta g(\cdot,\theta)g(\cdot,\theta)^\top+g(\cdot,\theta)\nabla_\theta g(\cdot,\theta)^\top\right] \\
    &\quad - G_{n,m,V}(\theta)S(\theta)^\top - S(\theta)G_{n,m,V}(\theta)^\top.
\end{align*}
Hence we have
$$\|\nabla_\theta \hat\Omega_{n,m,V}(\theta)\| \le
2\mathbb P_{n,m,V} \{\|\nabla_\theta g(\cdot,\theta)\|\,\|g(\cdot,\theta)\|\} +
2\|G_{n,m,V}(\theta)\|\,\|S(\theta)\|.$$
By Cauchy-Schwarz, Assumption \ref{A2} and Lemma \ref{lem:wlln}, we have $$\mathbb{E} [\sup_{\theta \in \mathcal{N}} \|\nabla_{\theta} g(., \theta)\| \sup_{\theta\in \mathcal{N}} \|g(\cdot, \theta)\|] < \infty.$$ Since $\sup_{\theta} \|S(\theta)\|$ is also $O_p(1)$ by Lemma \ref{lem:wlln}, $\nabla_{\theta} \hat{\Omega}_{n,m,V}(\theta)$ is $O_p(1)$.
Finally, the uniform boundedness for $\hat{\Omega}_{n,m,V}^{-1}(\theta)$ comes from Assumption \ref{A1} which asserts $\Omega_0$ is positive definite. Then continuity of $\Omega(\theta)$ and uniform LLN on $gg^{\top}$ and $g$ makes the smallest eigenvalues of $\hat{\Omega}_{n,m,V}(\theta)$ lower bounded by a positive constant. Hence $\sup_{\theta \in \mathcal{N}} \|\hat{\Omega}_{n,m,V}^{-1}(\theta)\| = O_p(1)$.
  
\end{proof}

\begin{lemma} \label{lem:ub-of-s}
With our choice of $r_n$ and the definition of the set $B_n:=\{\theta:\|\theta-\hat\theta_n\|\le r_n\}$, we have
$$\sup_{\theta \in B_n} \|S(\theta)\| = o_p(1).$$
\end{lemma}
\begin{proof}
    Since $\hat{\theta}_n \xrightarrow{p} \theta_0$, the ball $B_n$ lies in $\mathcal{N}$ w.h.p. For any $\theta \in B_n$, we have
    $$\|S(\theta)\| \leq \|S(\theta) - \bar g(\theta)\| + \|\bar g(\theta)-\bar g(\theta_0)\|.$$
    The first term is $o_p(1)$ by Dirichlet-weighted LLN on $\mathcal{N}$. By assumption \ref{A2}, we have
    $\sup_{\theta \in \mathcal{N}} \|E[\nabla_{\theta} g(x,\theta)]\| \leq \mathbb{E} [\sup_{\theta \in \mathcal{N}} \|\nabla_{\theta} g(x, \theta)\|] < \infty.$ This implies the second term is controlled by $L(\|\theta-\theta_0\|)$ for some constant $L>0$. For $\theta \in B_n$, we have $\|\theta-\theta_0\| \leq \|\theta -\hat{\theta}_n\|+ \|\hat{\theta}_n - \theta_0\| \leq r_n + \|\hat{\theta}_n - \theta_0\|$. It follows that
    $$\sup_{\theta \in B_n} \|S(\theta)\| \leq o_p(1) + L(r_n + \|\hat{\theta}_n - \theta_0\|)= o_p(1).$$
\end{proof}

\begin{lemma} \label{lem:strong-convexity}
Under Assumptions \ref{A1}--\ref{A4}, there exists $c>0$ such that
$$P\left\{ \inf_{\theta\in\mathcal N} \lambda_{\min}\bigl(\nabla_\theta^2 l_{n,m,V}(\theta)\bigr)\ge c \right\}\to 1.$$
Moreover, for $B_n=\{\theta:\|\theta-\hat\theta_n\|\le r_n\}$ with
$r_n\downarrow0$ and $r_n\sqrt n\to\infty$,
$$\sup_{\theta \in B_n} \|\nabla_{\theta}^2 l_{n,m,V}(\theta) - J_0\| = o_p(1).$$
\end{lemma}
\begin{proof}
    Define $J(\theta):=G(\theta)^\top\Omega(\theta)^{-1}G(\theta)$ and $J_{n,m,V}(\theta):= G_{n,m,V}(\theta)^{\top} \hat{\Omega}^{-1}_{n,m,V}(\theta)G_{n,m,V}(\theta)$. By Lemma \ref{lem:loss-2nd-d}, $\nabla_{\theta}^2 l_{n,m,V}(\theta) = J_{n,m,V}(\theta) + O(\|S(\theta)\|)$ uniformly on $\mathcal{N}$. By Lemma \ref{lem:wlln}, $$ \sup_{\theta\in\mathcal N}
\|S(\theta)-\bar g(\theta)\|=o_p(1) ,\qquad \sup_{\theta\in\mathcal N} \|J_{n,m,V}(\theta)-J(\theta)\|=o_p(1).$$ Additionally, assumption \ref{A1} gives $\inf_{\theta\in\mathcal N}\lambda_{\min}\{J(\theta)\}>0$. Since $\bar g(\theta_0)=0$ and $\bar g$ is continuous, $\mathcal N$ can be chosen small enough so that $\sup_{\theta\in\mathcal N}\|\bar g(\theta)\|$
is sufficiently small. Hence the $O(\|S(\theta)\|)$ term cannot alter the positive lower bound of $J(\theta)$, so the first claim follows.

To prove the second claim, recall that Lemma \ref{lem:ub-of-s} gives 
$$\sup_{\theta\in B_n}\|S(\theta)\|=o_p(1).$$
Given that $\nabla_{\theta}^2 l_{n,m,V}(\theta) = J_{n,m,V}(\theta) + O(\|S(\theta)\|)$ uniformly on $\mathcal{N}$, this implies 
$$\sup_{\theta\in B_n} \|\nabla_\theta^2 l_{n,m,V}(\theta)-J_{n,m,V}(\theta)\|=o_p(1).$$
Moreover, the weighted LLN in Lemma \ref{lem:wlln} gives $\sup_{\theta\in\mathcal N} \|J_{n,m,V}(\theta)-J(\theta)\|=o_p(1)$. Since $B_n$ shrinks to $\theta_0$ in probability and $J(\theta)$ is continuous at $\theta_0$, we also have $\sup_{\theta\in B_n}\|J(\theta)-J_0\|=o_p(1)$. The claim therefore follows by triangle inequality.
\end{proof}

\begin{lemma} \label{lem:dir-clt}
Let $\hat\theta_n$ be the standard ETEL estimator based on the observed
sample. Under Assumptions \ref{A1}--\ref{A4},
$$\sqrt n \left[ S_{n,m,V}(\hat\theta_n) -
    \mathbb{E}\{S_{n,m,V}(\hat\theta_n)\mid \mathcal{D}_n\} \right]
    \mid \mathcal{D}_n \overset{\mathbb P}{\rightsquigarrow} \mathcal N(0,\Omega_0).$$
Moreover,
$$\sqrt n\, \hat G_n(\hat\theta_n)^\top \hat\Omega_n(\hat\theta_n)^{-1}
    \mathbb{E}\{S_{n,m,V}(\hat\theta_n)\mid \mathcal{D}_n\} = o_p(1),
$$
where $ \hat G_n(\theta) :=\mathbb P_n\nabla_\theta g(\cdot,\theta),$ and
$$\hat\Omega_n(\theta):= \mathbb P_n\{g(\cdot,\theta)g(\cdot,\theta)^\top\}
    - \mathbb P_ng(\cdot,\theta)\mathbb P_ng(\cdot,\theta)^\top.$$
\end{lemma}
\begin{proof}
     We follow the same setup and notations as in Lemma \ref{lem:wlln}. In particular, we write  $\gamma_i \overset{\text{i.i.d.}}{\sim} \Gamma(1,1)$ for $i\in \{1, \ldots, n\}$ with $W_n:= \sum_{i=1}^n \gamma_i$, and $\gamma_{n+j} \overset{\text{i.i.d.}}{\sim} \Gamma(\alpha_n/m,1)$ for $j\in \{1, \ldots, m\}$ with $W_{\alpha_n}= \sum_{j=1}^m \gamma_{n+j}$. Define $$h_i:=g(x_i,\hat\theta_n), \quad  T_n:=\sum_{i=1}^n\omega_i h_i,
    \quad U_n:=\sum_{j=1}^m\pi_j g(x_j^*,\hat\theta_n).$$ 
    Then the representation in Lemma \ref{lem:wlln} gives
    $$ S_{n,m,V}(\hat\theta_n) = (1-\delta_n)T_n+\delta_n U_n.$$
    Given that  $\mathbb{E} \{S_{n,m,V}(\hat\theta_n)\mid \mathcal{D}_n\} =
    \frac{n}{n+\alpha_n} \mathbb P_n g(\cdot,\hat\theta_n) + \frac{\alpha_n}{n+\alpha_n} \mathbb E_{F_{\mathrm{AI}}} g(.,\hat\theta_n)$, and the auxiliary part is negligible at the $\sqrt{n}$ scale, we have
    $$ \sqrt n \left[S_{n,m,V}(\hat\theta_n)- \mathbb{E}\{S_{n,m,V}(\hat\theta_n)\mid \mathcal{D}_n\} \right] =
    \sqrt n \left[ T_n-\mathbb P_n g(\cdot,\hat\theta_n) \right]
    + o_{p}(1).$$

    We introduce random variables $\xi_i$, so that $\mathbb{E}[\xi_i]=0$ and $\mathbb{E}[\xi_i^2]=1$, such that $\xi_i = \gamma_i-1$. Since $ T_n =\frac{n^{-1}\sum_{i=1}^n\gamma_i h_i}{n^{-1}\sum_{i=1}^n\gamma_i}$, we can write its numerator as $\mathbb P_n g(\cdot,\hat{\theta}_n)+\frac1n\sum_{i=1}^n\xi_i h_i$, and its denominator as $1+\frac{1}{n}\sum_{i=1}^n \xi_i$. Hence we have
    $$\sqrt{n}(T_n(\hat{\theta}_n)-\mathbb{P}_n g(\cdot,\hat\theta_n)) = \frac{1}{\sqrt{n}}\sum_{i=1}^n \xi_i g_i(\hat{\theta}_n) + o_p(1).$$
    In particular, we have $\rm{Var}(\xi_i g_i(\hat{\theta}_n)|\mathcal{D}_n)= g_i(\hat{\theta}_n) g_i(\hat{\theta}_n)^{\top}$ and $\Omega_n(\hat{\theta}_n)= \mathbb{P}_n[g(\cdot, \hat{\theta}_n) g(\cdot, \hat{\theta}_n)^{\top}] \xrightarrow{p} \Omega_0.$ Since $\sup_{\theta} \mathbb{E}[\sup_{\theta} \|g(\cdot, \theta)\|^{2+\delta}] < \infty$ for some $\delta > 0$ by Assumption \ref{A2}, the Lyapunov condition holds. Hence by the Lindeberg-Feller CLT, we have 
    $$ \sqrt n \left[ S_{n,m,V}(\hat\theta_n) -
    \mathbb{E}\{S_{n,m,V}(\hat\theta_n)\mid \mathcal{D}_n\} \right]
    \mid \mathcal{D}_n \overset{\mathbb P}{\rightsquigarrow} \mathcal N(0,\Omega_0).$$

    To show the second part of the lemma, recall our score derivation in (\ref{eq:score-expansion}). By a similar argument for uniform weights, we have
    $$\nabla_{\theta} l_n(\theta) = \hat{G}_n(\theta)^\top \hat{\Omega}^{-1}_n(\theta) \mathbb{P}_n g(\cdot,\theta) + O(\|\mathbb{P}_n g(\cdot,\theta)\|^2). + O(\|\mathbb{P}_n g(\cdot,\theta)\|^2).$$ By the first-order condition for $\hat{\theta}_n$, we have $$\hat{G}_n(\hat{\theta}_n)^\top \hat{\Omega}^{-1}_n(\hat{\theta}_n) \mathbb{P}_n[g(\cdot, \hat{\theta}_n)] = O(\|\mathbb{P}_n g(\cdot,\hat{\theta}_n)\|^2) = O_p(n^{-1}).$$
    The claim is then immediate by noting $\mathbb{E}[S(\hat{\theta}_n)|\mathcal{D}_n] = \frac{n}{n+\alpha_n} \mathbb{P}_n g(\cdot, \hat{\theta}_n) + \frac{\alpha_n}{n+\alpha_n} \mathbb{E}_{\rm{AI}} g(\cdot, \hat{\theta}_n)$, and the AI part is negligible asymptotically by our choice of $\alpha_n$.

\end{proof}

\begin{lemma} \label{lem:etel-score}
Under Assumptions \ref{A1}--\ref{A4}, evaluating the weighted ETEL score at the standard ETEL estimator yields 
$$\nabla_{\theta}l_{n,m,V}(\hat{\theta}_n) = G_{n,m,V}(\hat{\theta}_n)^{\top} \hat{\Omega}_{n,m,V}^{-1}(\hat{\theta}_n) S(\hat{\theta}_n) + o_p(n^{-1/2}).$$
Moreover,
$$ \sqrt n\,\nabla_\theta l_{n,m,V}(\hat\theta_n) \mid \mathcal{D}_n
    \overset{\mathbb P}{\rightsquigarrow} \mathcal N(0,J_0).$$
 \end{lemma} 
\begin{proof}
    By our derivation of the score in the main Theorem \ref{thm:bootstrap}, we have
    $$\nabla_{\theta}l_{n,m,V}(\hat{\theta}_n) = G_{n,m,V}(\hat{\theta}_n)^{\top}\hat{\Omega}_{n,m,V}^{-1}(\hat{\theta}_n)S(\hat{\theta}_n) + O(\|S(\hat{\theta}_n)\|^2).$$
    We will first show $\|S(\hat{\theta}_n)\|$ is $O_p(n^{-1/2})$. To this end, we may decompose $S(\hat{\theta}_n)$ as follows:
    $$S(\hat{\theta}_n)=(S(\hat{\theta}_n)- \mathbb{E}[S(\hat{\theta}_n)| \mathcal{D}_n]) + \mathbb{E}[S(\hat{\theta}_n)| \mathcal{D}_n].$$
    By Lemma \ref{lem:dir-clt}, we have $\sqrt{n}(S(\hat{\theta}_n)- \mathbb{E}[S(\hat{\theta}_n)|\mathcal{D}_n]) \xrightarrow{d} N(0, \Omega_0)$. Additionally, 
    $$E[S(\hat{\theta}_n)| \mathcal{D}_n] = \mathbb{P}_n g(\cdot, \hat{\theta}_n) + o_p(n^{-1/2}) = O_p(n^{-1/2}).$$
     Hence $\|S(\hat{\theta}_n)\| = O_p(n^{-1/2})$, which implies $O(\|S(\hat{\theta}_n)\|^2)$ is $o_p(n^{-1/2})$. 
    
    Finally by continuity and LLN, we have $G_{n,m,V}(\hat{\theta}_n) \xrightarrow{p} G_0$ and $\hat{\Omega}_{n,m,V}^{-1}(\hat{\theta}_n) \xrightarrow{p} \Omega_0^{-1}$. An application of Slutsky's lemma yields $\sqrt{n} \nabla_{\theta} l_{n,m,V}(\hat{\theta}_n) \xrightarrow{d} \mathcal{N}(0, J_0)$, since $$\sqrt{n} \hat{G}(\hat{\theta}_n)\hat{\Omega}^{-1}(\hat{\theta}_n) \mathbb{E}[S(\hat{\theta}_n)|\mathcal{D}_n]$$ is $o_p(1)$ by the second part of Lemma \ref{lem:dir-clt}.
\end{proof}

\begin{lemma} \label{lem:localize-loss}
Under Assumptions \ref{A1}-\ref{A4}, $$\|\theta^*(V_n)-\hat\theta_n\| = O_p(n^{-1/2}).$$
\end{lemma}

 \begin{proof}
Since $\hat\theta_n\xrightarrow{p}\theta_0$ by standard ETEL consistency and $\theta^*(V_n)\xrightarrow{p}\theta_0$ by Lemma \ref{lem:global-consistency}, we know the line segment
$$\theta_t:=\hat\theta_n+t\{\theta^*(V_n)-\hat\theta_n\}, \qquad t\in[0,1],$$
must also lie in the fixed convex neighborhood $\mathcal{N}$ with probability $1-o(1)$.

By the first-order condition $\nabla_\theta l_{n,m,V}(\theta^*(V_n))=0$, we have
$$0 = \nabla_\theta l_{n,m,V}(\hat\theta_n) + \left[
\int_0^1 \nabla_\theta^2l_{n,m,V}(\theta_t)\,dt \right]
\{\theta^*(V_n)-\hat\theta_n\}.$$
By the first part of Lemma \ref{lem:strong-convexity}, the matrix in brackets has inverse with operator norm $O_p(1)$. It follows that
$$\|\theta^*(V_n)-\hat\theta_n\| \le O_p(1)\ \|\nabla_\theta l_{n,m,V}(\hat\theta_n)\| = O_p(n^{-1/2}).$$
The final equality follows from Lemma \ref{lem:etel-score}.
 \end{proof}

\subsection{Proof of Theorem \ref{thm:bootstrap}} \label{subsec:proof-bootstrap}

\begin{proof}
    For simplicity, write $g_k(\theta):= g(x_k, \theta)$, $S(\theta):= S_{n,m,V}(\theta)$, and $\psi(\eta, \theta):= \psi_{n,m,V}(\eta, \theta)$. We start by expanding $\psi(\lambda,\theta)$ at $\eta=0$ via Taylor expansion for $\theta \in \mathcal{N}$. Since
    \begin{align} \label{eq:hat-omega}
        \nabla_{\eta} \psi(\eta, \theta) &= \frac{\sum_k v_k e^{\eta^{\top} g_k(\theta)}g_k(\theta)}{\sum_k v_k e^{\eta^{\top} g_k(\theta)}},\notag \\
        \nabla^2_{\eta \eta} \psi(0, \theta) &= \sum_{k} v_k g_k(\theta) g_{k}(\theta)^{\top} - S(\theta)S(\theta)^{\top} := \hat{\Omega}_{n,m,V}(\theta),
    \end{align}
    we have 
    \begin{align} \label{eq:psi_t_expand}
        \psi(\lambda, \theta) &= \psi(0,\theta) + \nabla_{\eta} \psi(0,\theta)^{\top} \lambda(\theta) + \frac{1}{2} \lambda(\theta)^{\top} \hat{\Omega}_{n,m,V}(\theta) \lambda(\theta) + R_3(\theta) \notag \\
         &= S(\theta)^{\top} \lambda(\theta) + \frac{1}{2} \lambda(\theta)^{\top} \hat{\Omega}_{n,m,V}(\theta) \lambda(\theta) + R_3(\theta),
    \end{align}
    where $R_3(\theta) \leq \frac{1}{6} \ \| \nabla_{\eta}^3 \psi(\bar{\eta}, \theta) \| \|\lambda(\theta)\|^3$ for some $\bar{\eta}$ on the segment $[0, \lambda(\theta)]$. By the first part of Lemma \ref{lem:error}, we show the operator norm of $\| \nabla_{\eta}^3 \psi(\bar{\eta}, \theta) \|$ is uniformly bounded  and $\|\lambda (\theta)\| \leq C_1\|S(\theta)\|$ for $\|\eta\| \leq \eta_0$ and for some constant $C_1$ in part (ii) of Lemma \ref{lem:error}. Hence we have $R_3(\theta) = O(\|S(\theta)\|^3)$.

    To represent the dual objective $\lambda(\theta)$, we conduct additional Taylor expansion of $\nabla_{\eta} \psi(\lambda_{\theta}, \theta)$ at $\eta=0$. Observe that by the first-order condition, $\nabla_{\eta} \psi(\lambda_{\theta}, \theta)=0$, we have
    $$0=S(\theta)+\hat\Omega_{n,m,V}(\theta)\lambda(\theta)+r_\lambda(\theta), \qquad
        \|r_\lambda(\theta)\|\le C\|\lambda(\theta)\|^2.$$
     By Lemma \ref{lem:wlln}, we know $\hat{\Omega}_{n,m,V}(\theta) \xrightarrow{p} \Omega(\theta)$.  By Assumption \ref{A1}, and the standing choice of $\mathcal{N}$ (shrinking if needed), we have $\inf_{\theta \in \mathcal{N}} \lambda_{min} [\Omega(\theta)] > 0$. Hence with large $n$, $\hat{\Omega}_{n,m,V}(\theta)$ is uniformly invertible w.h.p. in $\mathcal{N}$. This implies
    \begin{align}  \label{eq:dual_t_expand}
        \lambda(\theta) = - \hat{\Omega}_{n,m,V}^{-1}(\theta) S(\theta) + \Delta_{\lambda}(\theta),
    \end{align}
    where $ \Delta_{\lambda}(\theta):= -  \hat{\Omega}_{n,m,V}^{-1}(\theta) r_{\lambda}(\theta)$. By part (ii) of Lemma \ref{lem:error}, the operator norm $\|\Delta_{\lambda}(\theta)\| \leq C_2 \|S(\theta)\|^2$ for some constant $C_2$. 
    
    We may rewrite the loss function $l_{n,m,V}(\theta)$ using our results in (\ref{eq:psi_t_expand}) and (\ref{eq:dual_t_expand}):
    \begin{align} \label{eq:taylor-loss}
    l_{n,m,V}(\theta) &=  \psi_{n,m,V}(\lambda(\theta), \theta) - \lambda(\theta)^{\top} S(\theta) \notag \\
    &= S(\theta)^{\top} \lambda(\theta) + \frac{1}{2} \lambda(\theta)^{\top} \hat{\Omega}_{n,m,V}(\theta) \lambda(\theta) + R_3(\theta) - \lambda(\theta)^{\top} S(\theta) \notag \\
    & = \frac{1}{2} S(\theta)^{\top} \hat{\Omega}_{n,m,V}^{-1}(\theta) S(\theta) + O(\|S(\theta)\|^3)
    \end{align}
    Define $G_{n,m,V}(\theta):= \sum_k v_k \nabla_{\theta} g(x_k, \theta)$ and $\dot{\Omega}:= \frac{\partial \hat{\Omega}_{n,m,V}(\theta)}{\partial \theta}$. Then differentiating the quadratic form above yields
    $$\nabla_{\theta}  \frac{1}{2} S(\theta)^{\top} \hat{\Omega}_{n,m,V}^{-1}(\theta) S(\theta) = G_{n,m,V}(\theta)^{\top} \hat{\Omega}_{n,m,V}^{-1}(\theta) S(\theta) -\frac{1}{2} S(\theta)^{\top} \hat{\Omega}_{n,m,V}^{-1}(\theta) \dot{\Omega} \hat{\Omega}_{n,m,V}^{-1}(\theta)S(\theta).$$
    Note that the second term is $O(\|S(\theta)\|^2)$, since $\sup_{\theta \in \mathcal{N}} \|\hat{\Omega}_{n,m,V}^{-1}(\theta)\| = O_p(1)$, and $\dot{\Omega}$ is uniformly $O_p(1)$ on $\mathcal{N}$ as detailed in Lemma \ref{lem:loss-2nd-d}. It follows that
    \begin{align} \label{eq:score-expansion}
        \nabla_{\theta} l_{n,m,V}(\theta) := G_{n,m,V}(\theta)^{\top} \hat{\Omega}_{n,m,V}^{-1}(\theta) S(\theta) + O(\|S(\theta)\|^2).
    \end{align}
    In Lemma \ref{lem:loss-2nd-d}, we further show under assumption \ref{A2}, the second derivative of the loss has the form 
    \begin{align} \label{eq:score-2nd-expansion}
       \nabla_{\theta}^2l_{n,m,V}(\theta) = G_{n,m,V}(\theta)^{\top} \hat{\Omega}_{n,m,V}^{-1}(\theta) G_{n,m,V}(\theta) + O(\|S(\theta)\|). 
    \end{align}
    Our remaining goal is to relate $\theta^*(V_n)$ to the ETEL estimator $\hat{\theta}_n$ by conducting a Taylor expansion of $\nabla_{\theta} l_{n,m,V}(\theta)$ at $\hat{\theta}_n$ and appealing to the first-order condition induced by $\theta^*(V_n)$. To this end, we fix a deterministic sequence $r_n \downarrow 0$ with $r_n \sqrt{n} \rightarrow \infty$, and consider a sequence of closed neighborhoods around $\hat{\theta}_n$: $$B_n:=\{\theta: \|\theta- \hat{\theta}_n\| \leq r_n\}.$$ In particular, $S(\theta)$ can be easily controlled in $B_n$ such that $\sup_{\theta \in B_n} \|S(\theta)\| = o_p(1)$ by Lemma \ref{lem:ub-of-s}. Since the standard ETEL estimator is consistent \cite{Schennach2007ETEL}, $B_n \subset \mathcal{N}$ with probability $1-o(1)$.  

    By Theorem \ref{thm:posterior-consistency}, $\theta^*(V_n)\xrightarrow{p}\theta_0$. Since the standard ETEL estimator is also consistent, $\hat{\theta}_n\xrightarrow{p}\theta_0$. Therefore, with probability tending to one, both $\theta^*(V_n)$ and $\hat\theta_n$ lie in $\mathcal N$. By Lemma \ref{lem:localize-loss}, we have $\|\theta^*(V_n)-\hat\theta_n\|=O_p(n^{-1/2})$. 
    Hence, by our choice of $r_n$ such that $r_n\sqrt n\to\infty$, $P\{\theta^*(V_n)\in B_n\}\to1$. It follows that with probability $1-o(1)$, the line segment
    $$\theta_t:=\hat\theta_n+t\{\theta^*(V_n)-\hat\theta_n\},
    \qquad t\in[0,1],$$
    is contained in $B_n$.
    
     By the mean-value theorem and the first-order condition induced by $\theta^*(V_n)$, we have
    $$0 = \nabla_{\theta} l_{n,m,V}(\theta^*(V_n)) = \nabla_{\theta} l_{n,m,V}(\hat{\theta}_n) + \left(\int_{0}^1 \nabla_{\theta}^2 l_{n,m,V}(\theta_t) dt \right)(\theta^*(V_n)-\hat{\theta}_n).$$
    To study the behavior of $\bar{J}_{n,m}:=\int_{0}^1 \nabla_{\theta}^2 l_{n,m,V}(\theta_t) dt $, we define $J_0:= G_0^{\top} \Omega_0^{-1} G_0$ and appeal to Lemma \ref{lem:strong-convexity}, which demonstrates that $\sup_{\theta \in B_n} \|\nabla_{\theta}^2 l_{n,m,V}(\theta) - J_0\| = o_p(1)$. It follows that $\bar{J}_{n,m} \xrightarrow{p} J_0$, which is invertible w.h.p. Hence we have
    $$\sqrt{n}(\theta^*(V_n)- \hat{\theta}_n) = - \bar{J}_{n,m}^{-1} \sqrt{n} \nabla_{\theta} l_{n,m,V}(\hat{\theta}_n).$$
    To conclude the proof, we note that $\sqrt{n} \nabla_{\theta} l_{n,m,V}(\hat{\theta}_n) \xrightarrow{d} \mathcal{N}(0, J_0)$ by Lemma \ref{lem:etel-score}. An application of Slutsky's lemma yields the desired result.
\end{proof}

\subsection{Proof of Theorem \ref{thm:b-non-vanishing}} \label{subsec:proof-non-vanishing}

\begin{proof}
    The structure of the proof closely parallels the vanishing prior case in Theorem \ref{thm:bootstrap}, and hence we only provide a sketch. The key difference from Theorem \ref{thm:bootstrap}  is that all population quantities are now defined under the mixed law $F_{\gamma}$ rather than under $F_0$. In consequence, our discussion will center around $\hat{\theta}_{n,\gamma}$ rather than on the standard ETEL estimator $\hat{\theta}_n$ under $F_0$. As shown in Lemma \ref{lem:wlln-prime}, the weighted Dirichlet LLN still holds under $F_{\gamma}$. Additionally, Lemma \ref{lem:clt-prime} establishes the corresponding Dirichlet CLT under the non-vanishing prior regime. Together, Lemma \ref{lem:wlln-prime} and Lemma \ref{lem:clt-prime} essentially ensure that the proof techniques used in Theorem \ref{thm:bootstrap} carry over with only minor modifications. 
    
    For notational simplicity, we write $A_n:=n+\alpha_n=(1+\gamma)n.$ By  Lemma \ref{lem:wlln-prime}, the  weighted Dirichlet laws converge uniformly to their $F_\gamma$ limits. Hence the analogue of Lemma \ref{lem:global-consistency} gives
    $$\theta^*(V_n)\xrightarrow{p}\theta_\gamma, \qquad \hat\theta_{n,\gamma}\xrightarrow{p}\theta_\gamma.$$
    It follows that both estimators lie in the local neighborhood around $\theta_\gamma$ on which the analogue of the local expansions in the proof of Theorem \ref{thm:bootstrap} holds. In particular, for $\theta$ in that neighborhood, we still have
      \begin{align*}
        l_{n,m,V}(\theta) &= \frac{1}{2} S(\theta)^{\top} \hat{\Omega}_{n,m,V}^{-1}(\theta) S(\theta) + O(\|S(\theta)\|^3), \\
         \nabla_{\theta} l_{n,m,V}(\theta) &= G_{n,m,V}(\theta)^{\top} \hat{\Omega}_{n,m,V}^{-1}(\theta) S(\theta) + O(\|S(\theta)\|^2), \\
         \nabla_{\theta}^2l_{n,m,V}(\theta) &= G_{n,m,V}(\theta)^{\top} \hat{\Omega}_{n,m,V}^{-1}(\theta) G_{n,m,V}(\theta) + O(\|S(\theta)\|).
    \end{align*}
    
    As before, to relate $\theta^*(V_n)$ with the mixture ETEL estimator $\hat{\theta}_{n,\gamma}$, we define 
    $$B_{n,\gamma} := \{\theta:\|\theta-\hat\theta_{n,\gamma}\|\le r_n\}, \qquad r_n\downarrow0,\qquad r_n\sqrt{A_n}\to\infty.$$
    The analogues of Lemmas \ref{lem:strong-convexity} and  \ref{lem:localize-loss} yield $\|\theta^*(V_n)-\hat\theta_{n,\gamma}\| = O_p(A_n^{-1/2})$, so that $P\{\theta^*(V_n)\in B_{n,\gamma}\} \to1$. We define the line segment:
    $$\theta_t = \hat\theta_{n,\gamma} + t\{\theta^*(V_n)-\hat\theta_{n,\gamma}\}, \qquad t\in[0,1].$$
    Then, with probability $1-o(1)$, $\theta_t\in B_{n,\gamma}$ for all $t\in[0,1]$. By the first-order condition and the mean-value theorem,
    $$0 = \nabla_\theta l_{n,m,V}(\theta^*(V_n)) = \nabla_\theta l_{n,m,V}(\hat\theta_{n,\gamma}) + \bar J_{n,m,\gamma} \{\theta^*(V_n)-\hat\theta_{n,\gamma}\},$$
    where $\bar J_{n,m,\gamma} := \int_0^1 \nabla_\theta^2 l_{n,m,V}(\theta_t) dt.$ Appealing to the analogue of Lemma \ref{lem:strong-convexity} gives $\bar J_{n,m,\gamma}\xrightarrow{p}J_\gamma.$ Finally, by Lemma \ref{lem:etel-score-mixed}, we have
    $$\sqrt{A_n}\{\theta^*(V_n)-\hat\theta_{n,\gamma}\}
\mid \mathcal D_{n,m_n}
\overset{\mathbb P}{\rightsquigarrow}
\mathcal N(0,J_\gamma^{-1}).$$
    The result follows since 
    $$\sqrt{A_n} \{\theta^*(V_n)-\hat\theta_{n,\gamma}\} = -\bar J_{n,m,\gamma}^{-1} \sqrt{A_n} \nabla_\theta l_{n,m,V}(\hat\theta_{n,\gamma}) \xrightarrow{d} N(0,J_\gamma^{-1}),$$
    by Slutsky's lemma.
\end{proof}

\begin{lemma}[Dirichlet Weighted LLN under Non-Vanishing Prior]\label{lem:wlln-prime}

Under Assumption \ref{A4prime} and the analogues of Assumptions \ref{A1}--\ref{A3} under $F_\gamma$,
$$\sup_{\theta\in\Theta} \left\| S(\theta)-\mu_\gamma(\theta)\right\|=o_p(1),$$
and
$$\sup_{\theta\in\Theta,\eta\in\Lambda} \left| \psi_{n,m,V}(\eta,\theta) - \log\mathbb E_{F_\gamma} \{e^{\eta^\top g(X,\theta)}\}\right|=o_p(1).$$
Moreover, the same weighted uniform law holds on the local neighborhood
of $\theta_\gamma$ for the derivative classes appearing in
Lemma \ref{lem:wlln}.
\end{lemma}
\begin{proof}
    The proof closely follows Lemma \ref{lem:wlln}. Here we sketch the proof for the first display. As in Lemma \ref{lem:wlln}, we  write $W_n:=\sum_{i=1}^n \gamma_i$, $W_{\alpha_n}:= \sum_{j=1}^{m_n} \gamma_{n+j}$, and $\delta_n:= \frac{W_{\alpha_n}}{W_n + W_{\alpha_n}}$. In particular, $S(\theta) := \sum_{k=1}^{n+m} v_k g(x_k, \theta)$ can be reparametrized as follows:
    $$S(\theta):= (1-\delta_n) \sum_{i=1}^n \omega_i g(x_i, \theta)+ \delta_n \sum_{j=1}^{m_n} \pi_j g(x_j^*, \theta),$$
    where $\omega \sim \text{Dirichlet}(\mathbbm{1}_n)$, and $\pi \sim \mathrm{Dirichlet}\{(\alpha_n/m_n)\mathbbm{1}_{m_n}\}.$. Unlike in the vanishing prior situation, we have $\delta_n \xrightarrow{p} \frac{\gamma}{1+\gamma}:=\delta_{\gamma}$ under strong AI prior. In Lemma \ref{lem:wlln}, we have shown $\sup_{\theta \in \mathcal{N}} \|\sum_{i=1}^n \omega_i g(x_i, \theta)- \mathbb{E}_{F_0}[g(x,\theta)]\| = o_p(1)$. The usual LLN also gives $\sup_{\theta \in \mathcal{N}} \|\mathbb{P}^*_m g(., \theta) - \mathbb{E}_{\rm{AI}}[g(x, \theta)]\| = o_p(1)$. It follows:
    \begin{align*}
        S(\theta)&= (1-\delta_n) \mathbb{P}_n g(\cdot, \theta) + \delta_n \mathbb{P}_m^{*} g(., \theta) + o_p(1)\\
        & =  (1-\delta_\gamma) \mathbb{E}_{F_0} [g(\cdot, \theta)] + \delta_{\gamma} \mathbb{E}_{\rm{AI}}[g(., \theta)] + o_p(1) \\ 
        & = \mathbb{E}_{F_{\gamma}}[g(., \theta)] + o_p(1),
    \end{align*}
    uniformly over $\theta \in \Theta$. This proves the desired LLN for the $\mathcal{H}_1$ case. The arguments for the other classes can proceed in a similar fashion.
\end{proof}

\begin{lemma} [Dirichlet Weighted CLT under Non-Vanishing Prior]\label{lem:clt-prime}
Let $\hat{\theta}_{n,\gamma}$ be the ETEL estimator under the mixture law based on the observed sample and augmented samples $\mathcal{D}_{n,m_n}$. Define $\bar S(\theta)
:= \mathbb E [S(\theta)\mid\mathcal D_{n,m_n}]$. Then conditional on $\mathcal{D}_{n,m_n}$, we have
\begin{align*}
   & \sqrt{n+\alpha_n} (S(\hat{\theta}_{n, \gamma}) - \bar{S}(\hat{\theta}_{n,\gamma}))\xrightarrow{d} \mathcal{N} (0, \Omega_{\gamma,0}).
\end{align*}
\end{lemma}
\begin{proof}
    Recall that we have defined $A_n= n+ \alpha_n$. Following the idea as in Lemma \ref{lem:dir-clt}, the Dirichlet parameter is defined as
    \[
    a_{n,k}
    =
    \begin{cases}
    1, & k=1,\ldots,n,\\
    \alpha_n/m_n, & k=n+1,\ldots,n+m_n.
    \end{cases}
    \]
    We again consider the Gamma representation of the Dirichlet distribution: let $Y_{n,k}\sim\Gamma(a_{n,k},1)$ and $V_k=\frac{Y_{n,k}}{\sum_{\ell}Y_{n,\ell}}$, then we have  $\mathbb{E}[V_k |\mathcal{D}_{n,m_n}] = \frac{a_{n,k}}{A_n} := p_{n,k}$. Set $z_{n,k}:=g(x_k,\hat\theta_{n,\gamma})$ and $\bar z_n:=\sum_k p_{n,k}z_{n,k}$. Then, we can write the scaled fluctuation of the Dirichlet-weighted average around its conditional mean as
    $$\sqrt{A_n}\sum_k(V_k-p_{n,k})z_{n,k} = \frac{A_n}{\sum_\ell Y_{n,\ell}} \left[
\frac{1}{\sqrt{A_n}}\sum_k (Y_{n,k}-a_{n,k})(z_{n,k}-\bar z_n) \right].$$
    Since $\sum_\ell Y_{n,\ell}/A_n\to_p1$, we will again apply the Lindeberg-Feller CLT to the independent triangular array inside the squared bracket. Given that $\max_k p_{n,k} \rightarrow 0$ and the weighted $(2+\delta)$ moment bound from the persistent-prior analogue of Assumption \ref{A2}, we can apply the Linderberg-Feller CLT to obtain
    $$\sqrt{A_n}\sum_k(V_k-p_{n,k})z_{n,k} \Rightarrow N(0,\Sigma_n),$$
    where $\Sigma_n = \sum_k p_{n,k} (z_{n,k}-\bar z_n)(z_{n,k}-\bar z_n)^{\top}.$
    Since $\hat\theta_{n,\gamma}\to_p\theta_\gamma$, we know $\Sigma_n\to_p\Omega_{\gamma,0}$, so the result follows.
   
\end{proof}

\begin{lemma} \label{lem:etel-score-mixed}
Evaluating the score at the mixed ETEL estimator yields 
$$\sqrt{n+\alpha_n} \nabla_{\theta} l_{n,m,V}(\hat{\theta}_{n,\gamma}) 
    \mid\mathcal D_{n,m_n}
    \overset{\mathbb P}{\rightsquigarrow}
    \mathcal N(0,J_\gamma),$$
where $J_\gamma = G_{\gamma,0}^\top\Omega_{\gamma,0}^{-1}G_{\gamma,0}. $
\end{lemma}

\begin{proof}
    Let $l_{n,\gamma}^{0}(\theta):=\ell(\theta;F_{n,\gamma})$ denote the deterministic ETEL criterion based on $F_{n,\gamma}$. Then by definition, we have $\nabla_\theta l_{n,\gamma}^{0}(\hat\theta_{n,\gamma})=0.$ Using the same score expansion as in the main proof of Theorem \ref{thm:bootstrap} (see \eqref{eq:score-expansion}), applied once to $l_{n,m,V}$ and once to $l_{n,\gamma}^{0}$, would yield
    $$\nabla_\theta l_{n,m,V}(\hat\theta_{n,\gamma})-\nabla_\theta l_{n,\gamma}^{0}(\hat\theta_{n,\gamma}) = G_{\gamma,0}^\top\Omega_{\gamma,0}^{-1} \{S(\hat\theta_{n,\gamma})-\bar S(\hat\theta_{n,\gamma}) + o_p((n+\alpha_n)^{-1/2})\}.$$
    Then multiplying by $\sqrt{n+\alpha_n}$ and applying Lemma $\ref{lem:clt-prime}$ gives the result.

\end{proof}

\section{Engel Curve Recovery: Additional Details}

\subsection{System prompt} \label{subsec:prompt-design}

We provide the exact system prompt used to generate synthetic food share alternatives in the Engel-curve experiment. In practice, we found conditional generation of the outcome ($y_i \mid x_i, z_i$) is often more stable than generating the entire joint distribution of $(x_i, y_i, z_i)$. We also found it useful to solicit the model’s own suggestions when designing the prompt, so that the generated samples better reflect the underlying data-generating process while remaining consistent with the substantive knowledge of domain experts.

\begin{lstlisting}[style=promptstyle, caption={Exact system prompt used in the Engel-curve experiment}, label={lst:Engel-system-prompt}]
You are a conservative conditional outcome generator for a scalar Engel-curve simulation.

You will receive simulated household rows with columns:
[id, log_total_expenditure, log_gross_earnings, food_share]

Each row represents a working-age couple without children from an expenditure survey.
Rows are sorted by log_total_expenditure and then log_gross_earnings only to make the broad pattern easier to see.

For each row, KEEP log_total_expenditure and log_gross_earnings FIXED and generate K alternative plausible food_share values.
Your main goal is to infer a smoothed local conditional distribution, not to reproduce the exact observed value for the same row.

Main objective:
- Produce conservative, smoothed alternatives of plausible food_share values.
- Underfit rather than overfit if uncertain.
- Shrink toward the local median pattern of nearby rows if uncertain.

Economic guidance:
- Engel's law implies that, on average, food_share tends to FALL as total expenditure rises.
- log_total_expenditure is the main predictor.
- log_gross_earnings may matter, but its direct effect beyond expenditure should be weaker, smoother, and secondary.
- Nearby rows in (log_total_expenditure, log_gross_earnings) should have similar central food_share values.
- Keep every alternative food_share between 0 and 1.
- Avoid extreme tails unless strongly supported by the overall pattern in the data.

Important:
- Use the observed food_share values only to learn the broad and local conditional pattern.
- Treat the observed food_share in the same row as noisy; do NOT simply copy it.
- Do NOT reproduce row-specific noise.
- The K alternatives should be ORDERED from low to high and should represent a conservative spread around the same conditional distribution.
- The middle value should be close to the conditional center.
- The outer values should be mild deviations around that center, not extreme outliers.

Output STRICT JSON only with schema:
{"rows": [[id, food_share_alt_1, food_share_alt_2, ..., food_share_alt_K], ...]}

Rules:
- Return exactly one output row for each input id.
- Keep ids unchanged.
- Every alternative must be either a finite number in [0,1].
- Do not return null.
- No text, no explanations, no markdown.
\end{lstlisting}

\section{Equity Return Predictions: More Details and Experiments} \label{sec:finance-appendix}

\subsection{Data Collection}

We collect financial headline data from July 1, 2025, to December 31, 2025, using RavenPack via the WRDS API. We focus on the top 40 U.S. firms based on their market capitalization (price $\times$ shares outstanding) as of June 30, 2025. Since our prediction task centers around overnight equity return signs, we restrict attention to headlines released between 4:00 p.m.\ on day $(t-1)$ and 9:00 a.m.\ on day $t$. This prediction setting has been extensively studied in the finance literature (see \cite{chen2022llmreturns, LopezLiraTang2023ChatGPTStocks}).

Our analysis is conducted at the firm-date level. We filter headlines based on the relevance score provided by RavenPack, retaining only those with scores of at least 50 (on a scale from 1 to 100) to ensure that the retained headlines are meaningfully associated with the corresponding firm. For data cleaning, all text is converted to lowercase; HTML entities are decoded and Unicode is normalized (NFKC) to standardize character representations; URLs and zero-width characters are removed; and all whitespace is collapsed to single spaces with leading and trailing spaces trimmed.

As the underlying news data are proprietary, we do not release the raw headlines. However, we provide complete Python code to replicate the data retrieval and preprocessing pipeline using the WRDS API.

\subsection{Prompting Procedure}

For each firm-date observation in the training sample (July–August), we query \texttt{gpt-5.2} $200$ times via the OpenAI API with temperature $0.8$ to obtain $200$ synthetic sentiment scores, denoted by $z$. These scores are then transformed into synthetic binary labels $y^*$ according to
\[
y^* \sim \text{Bernoulli}\!\left(\text{sigmoid}(p_0 + z)\right),
\]
where $p_0$ is chosen so that $\operatorname{sigmoid}(p_0)$ matches the
baseline positive rate, approximately $0.54$, in the training data. This transformation ensures that the generated labels are centered around the empirical base rate, thereby avoiding systematic deviations from the distribution of the observed data.

The exact prompt used to generate the synthetic sentiment scores is provided below. Similar to the Engel curve experiment, we found it useful to ask GenAI for tips to construct a good prompt.

\begin{lstlisting}[style=promptstyle, caption={Exact system prompt used in the equity return prediction experiment}, label={lst:finance-prompt}]
You are an annotator of overnight market-news tone for US equities.

Task
----
For each firm-date observation, read the supplied REAL headlines and aligned source codes.
Return n independent draws of a latent net overnight continous tone score z in [-2, 2].

Interpretation of z
-------------------
- +2.0 : very bullish / strong positive catalyst
- +1.0 : moderately positive
-  0.0 : mixed, neutral, or only weakly informative
- -1.0 : moderately negative
- -2.0 : very bearish / strong negative catalyst

Guidance
--------
1) Use only the supplied headlines and source codes.
2) You are not asked to infer the realized future return exactly; instead, score
   the news tone a plausible market participant might perceive overnight.
3) Many nights are mixed or weakly informative. Most draws should be near zero.
   Extreme values should be rare and reserved for clearly strong catalysts.
4) Administrative, exchange, filing, promotional, or routine press-release items
   are usually weaker evidence than independent reported news.
5) Analyst rating / price-target changes are moderate evidence.
6) Strong earnings/guidance surprises, major litigation/regulatory outcomes,
   financing stress, M&A, management shocks, outages, or clearly material product
   news can justify larger |z|.
7) Draws should vary modestly around your central judgment:
   - more dispersion when the evidence is mixed or ambiguous
   - tighter draws when the catalyst is clear
8) Do not output explanations.

\end{lstlisting}

\subsection{Additional Experiment: Generating Synthetic News with Synthetic Labels} \label{subsec:full-dgp-generation}

In this manuscript, we primarily use generative AI in a conditional manner: given observed covariates, we prompt the model to generate synthetic labels. In principle, one could instead use generative AI to model the entire data-generating process. However, this approach is substantially more computationally demanding, and in our empirical experiments we did not observe meaningful performance improvements from doing so.

Specifically, for the equity return prediction task, an alternative data augmentation strategy is to generate synthetic news together with synthetic labels, rather than generating labels alone. In our experiments, we generate approximately $24{,}000$ firm-date synthetic news bundles with corresponding synthetic labels, and perform inference by augmenting the training sample with these synthetic observations.

Table \ref{tab:finance_main_results_full} provides a full extension of Table \ref{tab:finance_main_results} in the main text. In particular, the rows labeled “GPT-ETEL (synthetic-label)” and “$\ell_2$-logistic” correspond exactly to the results reported in the main manuscript, allowing for direct comparison with additional augmentation strategies considered here. The table shows that modeling the full data-generating process via generative AI by jointly generating synthetic news and labels does not lead to further improvements over the synthetic-label approach alone. However, augmenting the training data with synthetic observations can meaningfully improve the baseline $\ell_2$-logistic model, suggesting that generative augmentation is still beneficial.

\begin{table}[!t]
\centering
\caption{\textbf{Test-set performance in overnight news prediction.}
Results are based on $500$ Monte Carlo replications.}
\label{tab:finance_main_results_full}
\small
\setlength{\tabcolsep}{7pt}
\renewcommand{\arraystretch}{1.15}
\begin{tabular}{lcccc}
\toprule
& \multicolumn{2}{c}{AUC} & \multicolumn{2}{c}{Accuracy} \\
\cmidrule(lr){2-3} \cmidrule(lr){4-5}
Method & Mean & SD & Mean & SD \\
\midrule
GPT-ETEL (synthetic-label)
& 0.5743 & 0.0114
& 0.5605 & 0.0110 \\

GPT-ETEL (synthetic news + label)
& 0.5734 & 0.0116
& 0.5585 & 0.0113 \\

$\ell_2$-logistic
& 0.5597 & 0.0163
& 0.5476 & 0.0143 \\

Synthetic-only ($\alpha=0$)
& 0.5184 & 0.0119
& 0.5238 & 0.0101 \\
\bottomrule
\end{tabular}
\end{table}

The “Synthetic-only ($\alpha=0$)” row corresponds to a model trained solely on synthetic news bundles using our ETEL procedure, without any real data. Although its performance is lower than the other methods, the fact that its AUC remains above $0.5$ indicates that the synthetic news prior alone contains nontrivial predictive signal, and using GPT to construct synthetic news remains a sensible choice. 

Overall, these results suggest that conditional synthetic label generation is often sufficient to capture the benefits of generative-AI augmentation in this setting. Understanding when modeling the full data-generating process via generative AI can provide additional gains remains an interesting direction for future research.

\section{ATE Density Plot} 

\begin{figure}[t]          
  \centering
  \includegraphics[width=0.7\textwidth]{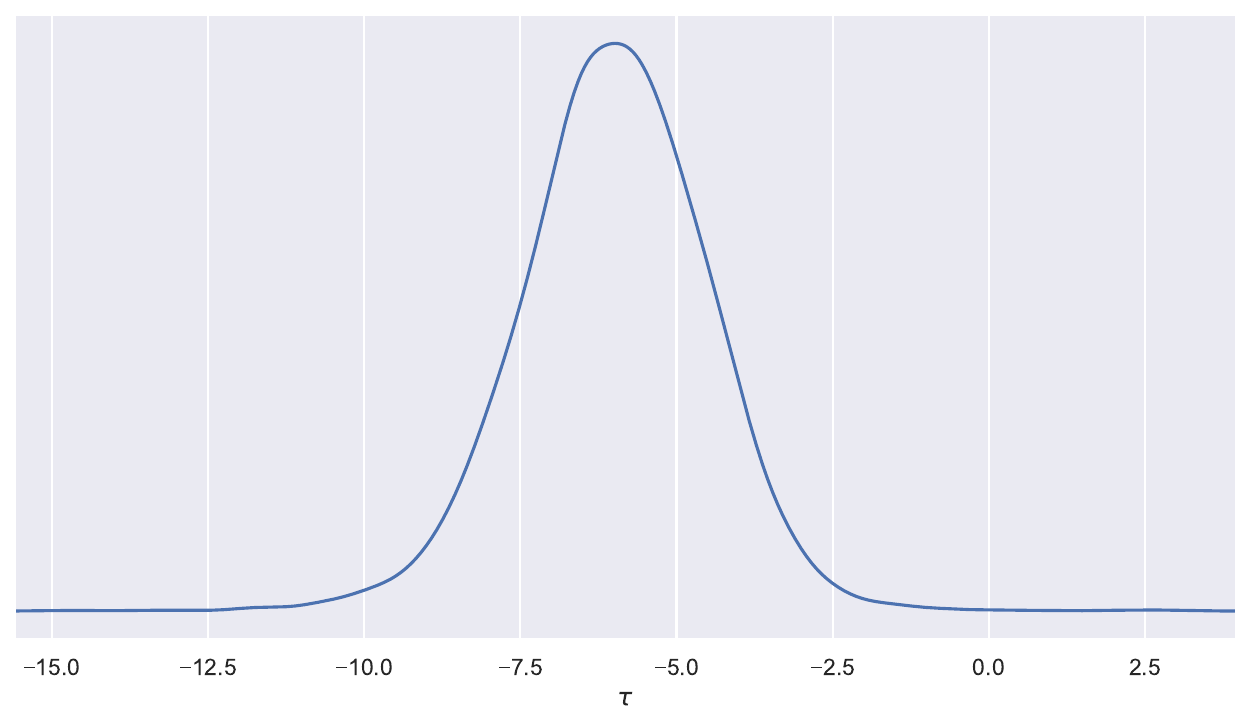} 
  \caption{Posterior distribution of the ATE on subsequent annual earnings of a substantial lottery win: ETEL bootstrap ($10,000$ draws).}
  \label{fig:tau}
\end{figure}

See Figure \ref{fig:tau} for posterior density estimation for the ATE experiment.
\end{document}